\documentclass[aos, preprint]{imsart}

\RequirePackage[OT1]{fontenc}
\RequirePackage{amsthm,amsmath}
\usepackage{array}
\usepackage{mathrsfs}
\usepackage{amssymb,bbm}
\usepackage{amsfonts}
\usepackage{amsmath,amssymb}
\usepackage{graphicx}
\usepackage{amscd}
\usepackage{color}
\usepackage{mathrsfs}
\usepackage{latexsym,bm}
\usepackage{array}
\usepackage{threeparttable}

\usepackage{threeparttable}
\usepackage{rotating}
\usepackage{booktabs}

\usepackage{rotating}

\usepackage[plain,noend]{algorithm2e}

\DeclareFontFamily{U}{mathx}{\hyphenchar\font45}
\DeclareFontShape{U}{mathx}{m}{n}{
      <5> <6> <7> <8> <9> <10>
      <10.95> <12> <14.4> <17.28> <20.74> <24.88>
      mathx10
      }{}
\DeclareSymbolFont{mathx}{U}{mathx}{m}{n}
\DeclareMathAccent{\widecheck}{\mathalpha}{mathx}{"71}

\DeclareMathAccent{\widecheck}{\mathalpha}{mathx}{"71}

\startlocaldefs \numberwithin{equation}{section} \theoremstyle{it}
\newtheorem{thm}{Theorem}[section]
\newtheorem{lem}{Lemma}[section]
\newtheorem{ass}{Assumption}[section]

\endlocaldefs

\newcommand{\f}{\frac}
\newcommand{\p}{\partial}

\begin{document}

\begin{frontmatter}
\title{New HSIC-based tests for independence between two stationary multivariate time series}


\begin{aug}
\author{\fnms{Guochang} \snm{Wang}\thanksref{t1,m1}\ead[label=e1]{twanggc@jnu.edu.cn}},
\author{\fnms{Wai Keung} \snm{Li}\thanksref{t2,m2}\ead[label=e2]{hrntlwk@hku.hk}}
\and
\author{\fnms{Ke} \snm{Zhu}\thanksref{t3,m2}
\ead[label=e3]{mazhuke@hku.hk}
\ead[label=u1,url]{http://www.foo.com}}

\thankstext{t1}{Supported in part by National Natural Science Foundation of China (No.11501248).}
\thankstext{t2}{Supported in part by Research Grants Council of the Hong Kong SAR Government (GRF
grant HKU17303315).}
\thankstext{t3}{Supported in part by National Natural Science Foundation of China (No.11571348, 11371354, 11690014, 11731015 and
71532013).}

\affiliation{Jinan University\thanksmark{m1} and The University of Hong Kong\thanksmark{m2}}
%
%

\address{Jinan University\\
College of Economics\\
Guangzhou, China\\
\printead{e1}\\
}

\address{The University of Hong Kong\\
Department of Statistics \& Actuarial Science\\
Hong Kong\\
\printead{e2}\\
\phantom{E-mail:\ }\printead*{e3}}
\end{aug}

\begin{abstract}
This paper proposes some novel one-sided omnibus tests for independence between two multivariate stationary time series. These new tests apply the Hilbert-Schmidt independence criterion
(HSIC) to test the independence between the innovations of both time series.
Under regular conditions, the limiting null distributions of
our HSIC-based tests are established. Next, our HSIC-based tests are shown to be consistent.
Moreover, a residual bootstrap method is used to obtain the critical values for our HSIC-based tests, and its validity is justified.
Compared
with the existing cross-correlation-based tests for linear dependence, our tests
examine the general (including both linear and non-linear) dependence to give investigators more complete
information on the causal relationship between two multivariate time series.
The merits of our tests are illustrated by some simulation results and a real example.
\end{abstract}


\begin{keyword}
\kwd{Hilbert-Schmidt independence
criterion; multivariate time series models; non-linear dependence; residual bootstrap; testing for independence} 
\end{keyword}

\end{frontmatter}

\section{Introduction}

Before applying any sophisticated method to describe relationships
between two time series, it is important to check whether they are
independent or not. If they are dependent, causal analysis
techniques, such as copula and multivariate modeling, can be used to
investigate the relationship between them, and this may lead to
interesting insights or effective predictive models; otherwise, one
should analyze them using two independent parsimonious models; see, e.g., Pierce
(1977), Schwert (1979), Hong (2001a), Lee and Long (2009), Shao (2009), and Tchahou and Duchesne (2013) for
many empirical examples in this context.

Most of the existing methods for testing the independence between two multivariate time series models use a measure based on
cross-correlations. Specifically, they aim to check whether the sample cross-correlations of model residuals,
 up to either certain fixed lag or all valid lags, are significantly different from zeros. The former includes the portmanteau tests (Cheung and Ng, 1996; El Himdi and Roy, 1997; Pham et al. 2003; Hallin and Saidi, 2005 and 2007; Robbins and Fisher, 2015), and the latter
with the aid of kernel smooth  technique falls in the category of spectral tests (Hong, 2001a and 2001b; Bouhaddioui and Roy, 2006).
It must be noted that the idea of using the cross-correlations is a natural extension of the pioneered studies in
Haugh (1976) and Hong (1996) for univariate time
series models, but in many circumstances it only suffices to
convey evidence of uncorrelatedness rather than independence.

Generally speaking, all of the aforementioned tests
are designed for investigating the linear dependence
(i.e., the cross-correlation in the mean, variance or higher moments) between two
model residuals, and hence they could exhibit a lack of power
in detecting the non-linear dependence structure. A significant body of research so far has
documented the non-linear dependence relationship among a myriad of economic fundamentals; see, e.g.,
Hiemstra and Jones (1994), Wang et al. (2013), Choudhry et al. (2016), and Diks and Wolski (2016) to name a few.
However, less attempts have been made in the literature to account for both linear and nonlinear dependence
structure, which shall be two parallel important characteristics to be tested.

To examine the general dependence structure, a direct measure on independence is expected for testing purpose. In the last decade,
the Hilbert-Schmidt independence
criterion (HSIC) in Gretton et al. (2005) has been extensively used
in many fields. Some inspiring works in one- or two-sample independence tests via HSIC
include Gretton et al. (2008) and Gretton and Gy\"{o}rfi (2010) for observable i.i.d. data,
and Zhang et al. (2009), Zhou (2012) and Fokianos and Pitsillou (2017) for observable dependent or time series data.
The last two instead applied the distance
covariance (DC) in Sz\'{e}kely et al. (2007), while Sejdinovic et al. (2013) showed that HSIC and DC are equivalent.
When the data are un-observable and derived from a fitted statistical model (e.g., the estimated model innovations), the estimation
effect has to be taken into account. The original procedure based on HSIC or DC will no longer be valid, and a modification of the above procedure
has to be derived for testing purpose.
By now, very little work has been done in this context. Two exceptions are
Sen and Sen (2014) and Davis et al. (2016) for one-sample independence tests; the former focused on the regression model with
independent covariates, and the latter considered the vector AR models but without providing a rigorous way to obtain
the critical values of the related test.

This paper proposes some novel one-sided tests for the independence
between two stationary multivariate time series.
These new tests apply the HSIC to examine the
independence between the un-observable innovation vectors of both time series.
Among them, the single HSIC-based test is tailored to
detect the general dependence between these two innovation vectors
at a specific lag $m$, and the joint HSIC-based test is designed for this purpose up to certain lag $M$.
Under regular conditions, the limiting null distributions of
our HSIC-based tests are established. Next, our HSIC-based tests are shown to be consistent.
Moreover, a residual bootstrap method is used to obtain the critical values for our HSIC-based tests, and its validity is justified.
Our methodologies are applicable
for the general specifications of the time series models driven by i.i.d. innovations.
By choosing different lags,
our new tests can give investigators more complete
information on the general (including both linear and non-linear) dependence
relationship between two time series.
Finally, the importance of our HSIC-based tests is illustrated by some simulation results and a real example.

This paper is organized as follows. Section 2 introduces our
HSIC-based test statistics and some technical assumptions. Section 3
studies the asymptotic properties of our HSIC-based tests. A residual bootstrap method
is provided in Section 4. Simulation
results are reported in Section 5. One real example is presented in
Section 6. Concluding remarks are offered in Section 7. The proofs
are provided in the Appendix.

Throughout the paper, $\mathcal{R}=(-\infty,\infty)$, $C$ is a
generic constant, $I_{s}$ is the $s\times s$ identity matrix,
$1_{s}$ is the
 $s\times1$ vector of ones, $\otimes$ is the Kronecker product,  $A^{T}$ is the transpose  of matrix $A$,
 $\|A\|$ is the Euclidean norm of
 matrix $A$,  $vec(A)$ is the vectorization of $A$,
 $vech(A)$ is the half vectorization of $A$,
 $D(A)$ is the diagonal matrix whose main diagonal is the main diagonal of
matrix $A$,
 $\partial_{x}h$ denotes the partial derivative with respect to $x$ for any function $h(x,y,\cdots)$,
 $o_{p}(1)(O_{p}(1))$ denotes a sequence of random numbers
 converging to zero (bounded) in probability,
 ``$\to_{d}$'' denotes convergence in distribution,
 and ``$\to_{p}$'' denotes convergence in probability.


\section{The HSIC-based test statistics}

\subsection{Review of the Hilbert-Schmidt Independence Criterion}

In this subsection, we briefly review the Hilbert-Schmidt
independence criterion (HSIC) for testing the independence of two
random vectors; see, e.g., Gretton et al. (2005) and Gretton et al.
(2008) for more details.

Let $\mathcal{U}$ be a metric space, and
$k: \mathcal{U}\times \mathcal{U}\to \mathcal{R}$ be a symmetric and positive definite (i.e., $\sum_{i,j}c_{i}c_{j}k(x_{i},x_{j})\geq0$ for all
$c_{i}\in\mathcal{R}$) kernel
function. There exists a Hilbert space $\mathcal{H}$ (called {\it
Reproducing Kernel Hilbert Space} (RKHS)) of functions $f: \mathcal{U}\to \mathcal{R}$ with
inner product $\langle \cdot,\cdot \rangle$ such that
\begin{flalign}
(i)&\,\, k(u,\cdot)\in\mathcal{H},\,\,\,\,\,\,\forall u\in\mathcal{U};\label{2.1}\\
(ii)&\,\,\langle f,k(u,\cdot)\rangle=f(u),\,\,\,\,\,\forall
f\in\mathcal{H} \mbox{ and }\forall u\in\mathcal{U}.\label{2.2}
\end{flalign}
For any Borel probability measure $P$ defined on $\mathcal{U}$, its {\it mean element}
$\mu[P]\in\mathcal{H}$ is defined as follows:
\begin{flalign}
E[f(U)]=\langle f,\mu[P]\rangle,\,\,\,\,\forall
f\in\mathcal{H},\label{2.3}
\end{flalign}
where the random variable $U\sim P$. From (\ref{2.2})-(\ref{2.3}), we have
$\mu[P](u)=\langle k(\cdot,u),\mu[P]\rangle=E[k(U,u)]$.
Furthermore, we say that $\mathcal{H}$ is {\it characteristic} if
and only if the map $ P\to \mu[P] $ is injective on the space
$\mathcal{P}:=\{P:\int_{\mathcal{U}}k(u,u)dP(u)<\infty\}$.

Likewise, let $\mathcal{G}$ be a second RKHS on a metric space
$\mathcal{V}$ with kernel $l$. Let $P_{uv}$ be a Borel probability
measure defined on $\mathcal{U}\times\mathcal{V}$, and let $P_{u}$
and $P_{v}$ denote the respective marginal distributions on
$\mathcal{U}$ and $\mathcal{V}$, respectively. Assume that
\begin{flalign}
E[k(U,U)]<\infty\,\,\,\,\,\mbox{ and
}\,\,\,\,\,E[l(V,V)]<\infty,\label{2.4}
\end{flalign}
where the random variable $(U,V)\sim P_{uv}$. The HSIC of
$P_{uv}$ is defined as
\begin{flalign*}
\Pi(U,V):&=E_{U,V}E_{U',V'}[k(U,U')l(V,V')]+E_{U}E_{U'}E_{V}E_{V'}[k(U,U')l(V,V')]\nonumber\\
&\quad\quad\quad-2E_{U,V}E_{U'}E_{V'}[k(U,U')l(V,V')],
\end{flalign*}
where $(U',V')$ is an i.i.d. copy of $(U,V)$, and $E_{\xi,\zeta}$ (or $E_{\xi}$) denotes the expectation over $(\xi,\zeta)$ (or $\xi$).
Following  Sejdinovic et al. (2013), if (\ref{2.4})
holds and both $\mathcal{H}$ and $\mathcal{G}$ are characteristic,
then
$$\Pi(U,V)=0\,\,\,\,\mbox{ if and only if }\,\,\,\, P_{uv}=P_{u}\times P_{v}.$$
Therefore, we can test the independence of $U$
and $V$ by examining whether $\Pi(U,V)$ is significantly different
from zero.

Suppose the samples $\{(U_{i},V_{i})\}_{i=1}^{n}$ are
from $P_{uv}$. Following  Gretton et al. (2005), the empirical estimator of
$\Pi(U,V)$ is
\begin{flalign}
\Pi_{n}&=\frac{1}{n^{2}}\sum_{i,j}k_{ij}l_{ij}+\frac{1}{n^{4}}\sum_{i,j,q,r}k_{ij}l_{qr}-\frac{2}{n^{3}}\sum_{i,j,q}k_{ij}l_{iq}\label{2.5}\\
&=\frac{1}{n^{2}}\mbox{trace}(KHLH),\label{2.6}
\end{flalign}
where $k_{ij}=k(U_{i},U_{j})$, $l_{ij}=l(V_{i},V_{j})$, $K=(k_{ij})$
and $L=(l_{ij})$ are $n\times n$ matrices with entries $k_{ij}$ and
$l_{ij}$, respectively, and $H=I_{n}-(1_{n}1_{n}^{T})/n$.
Here, each index of the summation $\sum$
is taken from 1 to $n$.
If $\{(U_{i},V_{i})\}_{i=1}^{n}$ are i.i.d. samples,
Gretton et al. (2005) showed that $\Pi_{n}$ is a
consistent estimator of $\Pi(U,V)$.

In order to compute $\Pi_{n}$, we need to choose the kernel
functions $k$ and $l$. In the sequel, we assume $\mathcal{U}=\mathcal{R}^{\kappa_{1}}$ and $\mathcal{V}= \mathcal{R}^{\kappa_{2}}$
for two positive integers $\kappa_{1}$ and $\kappa_{2}$.
Then, some well known choices (see Peters, 2008; Zhang et al. 2017)  for $k$ (or $l$) are given below:
\begin{flalign*}
[\mbox{Gaussian kernel}]&:
k(u,u')=\exp\left(-\frac{\|u-u'\|^{2}}{2\sigma^{2}}\right)\,\,\,\mbox{for some}\,\,\,\sigma>0;\\
[\mbox{Laplace kernel}]&:
k(u,u')=\exp\left(-\frac{\|u-u'\|}{\sigma}\right)\,\,\,\mbox{for
some}\,\,\,\sigma>0;\\
[\mbox{Inverse multi-quadratics kernel}]&:
k(u,u')=\frac{1}{(\beta+\|u-u'\|)^{\alpha}}\,\,\,\mbox{for
some}\,\,\,\alpha,\beta>0;\\
{\bf[\mbox{Fractional Brownian motion kernel}]}&:
k(u,u')=\frac{1}{2}(\|u\|^{2h}+\|u'\|^{2h}-\|u-u'\|^{2h}),\\
 &\,\,\,\,\,\,\,\mbox{for some}\,\,\,0<h<1.
\end{flalign*}

We shall highlight that the HSIC is easy-to-implement in multivariate cases, since
the computation cost of $\Pi_{n}$ is $O(n^{2})$
regardless of the dimensions of $U$ and $V$, and
many softwares can calculate (\ref{2.6}) very fast.


\subsection{Test statistics}
Consider two multivariate time series $Y_{1t}$ and $Y_{2t}$, where
$Y_{1t}\in\mathcal{R}^{d_{1}}$ and $Y_{2t}\in\mathcal{R}^{d_{2}}$.
Assume that each $Y_{st}$ ($s=1$ or $2$ hereafter) admits the
following specification:
\begin{flalign}
Y_{st}=f_{s}(I_{st-1},\theta_{s0},\eta_{st}),\label{2.8}
\end{flalign}
where $I_{st}=(Y_{st}^{T},Y_{st-1}^{T},\cdots)^{T}\in\mathcal{R}^{\infty}$ is
the information set at time $t$, $\theta_{s0}\in\mathcal{R}^{p_{s}}$
is the true but unknown parameter value of model (\ref{2.8}),
$\eta_{st}\in\mathcal{R}^{d_{s}}$ is a sequence of i.i.d.
innovations such that $\eta_{st}$ and $\mathcal{F}_{st-1}$ are
independent,
$\mathcal{F}_{st}:=\sigma(I_{st})$ is a sigma-field, and $f_{s}:
\mathcal{R}^{\infty}\times\mathcal{R}^{p_{s}}\times\mathcal{R}^{d_{s}}\to
\mathcal{R}^{d_{s}}$ is a known measurable function. Model
(\ref{2.8}) is rich enough to cover many often used models, e.g.,
the vector AR model in Sim (1980), the BEKK model in Engle and
Kroner (1995), the dynamic correlation model in Tse and Tsui (2002),
and the vector ARMA-GARCH model in Ling and McAleer (2003) to name a
few; see also L\"{u}tkepohl (2005), Bauwens et al. (2006),
Silvennoinen and Ter\"{a}svirta (2008), Francq  and Zako\"{i}an (2010), and Tsay (2014) for surveys.

Model (\ref{2.8}) ensures that each $Y_{st}$ admits a
dynamical system generated by the innovation sequence
$\{\eta_{st}\}$. A practical question is whether either one of the
dynamical systems should include the information from the other one,
and this is equivalent to testing the null hypothesis:
\begin{flalign}\label{null_1}
H_{0}: \{\eta_{1t}\} \mbox{ and } \{\eta_{2t}\} \mbox{ are
independent.}
\end{flalign}
If $H_{0}$ is accepted, we can separately study
these two systems; otherwise, we may use the information of one
system to get a better prediction of the other system.
Let $m$ be a given integer. Most of the conventional testing methods for $H_{0}$ in (\ref{null_1})
aim to detect the linear dependence between $\eta_{1t}$ and $\eta_{2t+m}$ (or their higher moments)
via their cross-correlations. Below,
we apply HSIC to examine the general dependence between $\eta_{1t}$ and $\eta_{2t+m}$.

%

To introduce our HSIC-based tests, we need some more notations. Let $\theta_{s}=(\theta_{s1},\theta_{s2},$ $\cdots,\theta_{sp_{s}})\in\Theta_{s}\subset\mathcal{R}^{p_{s}}$ be the unknown parameter of
model (\ref{2.8}), where $\Theta_{s}$ is a compact parametric space.
Assume that $\theta_{s0}$ is an interior point of $\Theta_{s}$, and $Y_{st}$ admits a causal
representation, i.e.,
\begin{flalign}\label{causal_rep}
\eta_{st}=g_{s}(Y_{st},I_{st-1},\theta_{s0}),
\end{flalign}
where $g_{s}:
\mathcal{R}^{d_{s}}\times\mathcal{R}^{\infty}\times\mathcal{R}^{p_{s}}\to\mathcal{R}^{d_{s}}$
is a measurable function. Moreover, based on the observations $\{Y_{st}\}_{t=1}^{n}$ and (possibly)  some assumed initial values,
we let
\begin{flalign}\label{residual}
\widehat{\eta}_{st}:=g_{s}(Y_{st},\widehat{I}_{st-1},\widehat{\theta}_{sn})
\end{flalign}
be the residual of model (\ref{2.8}), where
 $\widehat{\theta}_{sn}$ is an estimator of
$\theta_{s0}$, and $\widehat{I}_{st}$ is the observed information set up to time $t$.


As for (\ref{2.5})-(\ref{2.6}), our single HSIC-based
test statistic on $\widehat{\eta}_{1t}$ and $\widehat{\eta}_{2t+m}$ is
\begin{flalign}
S_{1n}(m):=\Pi(\widehat{\eta}_{1t},\widehat{\eta}_{2t+m})
&=\frac{1}{N^{2}}\sum_{i,j}\widehat{k}_{ij}\widehat{l}_{ij}+\frac{1}{N^{4}}\sum_{i,j,q,r}\widehat{k}_{ij}\widehat{l}_{qr}
-\frac{2}{N^{3}}\sum_{i,j,q}\widehat{k}_{ij}\widehat{l}_{iq}\nonumber\\
&=\frac{1}{N^{2}}\mbox{trace}(\widehat{K}H\widehat{L}H)\label{2.10}
\end{flalign}
for $m\geq0$, where $\widehat{k}_{ij}=k(\widehat{\eta}_{1i},\widehat{\eta}_{1j})$,
$\widehat{l}_{ij}=l(\widehat{\eta}_{2i+m},\widehat{\eta}_{2j+m})$, and
$\widehat{K}=(\widehat{k}_{ij})$ and
$\widehat{L}=(\widehat{l}_{ij})$ are $N\times N$ matrices with
entries $\widehat{k}_{ij}$ and $\widehat{l}_{ij}$, respectively.
Here, the effective sample size $N=n-m$, and each index of the summation
is taken from $1$ to $N$.
Likewise, our single HSIC-based
test statistic on $\widehat{\eta}_{1t+m}$ and $\widehat{\eta}_{2t}$ is
\begin{flalign}
S_{2n}(m):=\Pi(\widehat{\eta}_{1t+m},\widehat{\eta}_{2t})\label{2.11}
\end{flalign}
for $m\geq0$.
Clearly, $S_{1n}(0)=S_{2n}(0)$.

With the help of the single HSIC-based test statistics, we can further define the joint HSIC-based test statistics as follows:
\begin{flalign}
J_{1n}(M):=\sum_{m=0}^{M}S_{1n}(m)\,\,\,\mbox{ and }
\,\,\,J_{2n}(M):=\sum_{m=0}^{M}S_{2n}(m)\label{2.12}
\end{flalign}
for some specified integer $M\geq0$. The joint test statistic
$J_{1n}(M)$ or $J_{2n}(M)$ can detect the general dependence structure of
two innovations up to certain lag $M$, while
the single test statistic $S_{1n}(m)$ or $S_{2n}(m)$ is used to examine the general dependence structure of
two innovations at a specific lag $m$.


\section{Asymptotic theory} This section studies the asymptotics of
our HSIC-based test statistics $S_{1n}(m)$ and $J_{1n}(M)$.
The asymptotics of $S_{2n}(m)$ and $J_{2n}(M)$ can be derived similarly, and hence the details
are omitted for simplicity.

\subsection{Technical conditions}
To derive our asymptotic theory, the following assumptions are needed.

\begin{ass}
$Y_{st}$ is strictly stationary and ergodic.\label{asm2.1}
\end{ass}

\begin{ass}\label{asm2.2}
(i) The function $g_{st}(\theta_{s}):=g_{s}(Y_{st},I_{st-1},\theta_{s})$ satisfies that
\begin{flalign*}
&E\left[\sup_{\theta_{s}}\left\|\frac{\p
g_{st}(\theta_{s})}{\p\theta_{si}}\right\|\right]^{2}<\infty,\,\,\,\,\,\,\,\,\,\,
E\left[\sup_{\theta_{s}}\left\|\frac{\p^{2}
g_{st}(\theta_{s})}{\p\theta_{si}\p\theta_{sj}}\right\|\right]^{2}<\infty,\\
&\,\,\,\,\,\,\,\mbox{and}\,\,\,\,\,\,\,\,\,\,
E\left[\sup_{\theta_{s}}\left\|\frac{\p^{3}
g_{st}(\theta_{s})}{\p\theta_{si}\p\theta_{sj}\p\theta_{sq}}\right\|\right]^{2}<\infty,
\end{flalign*}
for any $i$, $j$, $q\in\{1,\cdots,p_{s}\}$, where $g_{s}$ is defined as in (\ref{causal_rep}).

(ii) $\sum_{j=0}^{\infty} \beta_{\eta}(j)^{c/(2+c)}<\infty$ for some $c>0$, where
$\beta_{\eta}(j)$ is the $\beta$-mixing coefficient of $\{(\eta_{1t}^{T},\eta_{2t}^{T})^{T}\}$.
\end{ass}

\begin{ass}
The estimator $\widehat{\theta}_{sn}$ given in (\ref{residual}) satisfies that
\begin{flalign}
\sqrt{n}(\widehat{\theta}_{sn}-\theta_{s0})&=\frac{1}{\sqrt{n}}\sum_{t}\pi_{s}(Y_{st},I_{st-1},\theta_{s0})+o_{p}(1)\nonumber\\
&=:
\frac{1}{\sqrt{n}}\sum_{t}\pi_{st}+o_{p}(1),\label{2.9}
\end{flalign}
where
$\pi_{s}:
\mathcal{R}^{d_{s}}\times\mathcal{R}^{\infty}\times\mathcal{R}^{p_{s}}\to\mathcal{R}^{p_{s}}$
is a measurable function,
$E(\pi_{st}|\mathcal{F}_{st-1})=0$, and
$E\|\pi_{st}\|^{2}<\infty$.\label{asm2.3}
\end{ass}

\begin{ass}
For $\widehat{R}_{st}(\theta_{s}):=\widehat{g}_{st}(\theta_{s})-
g_{st}(\theta_{s})$,
$$\sum_{t}\sup_{\theta_{s}}\|\widehat{R}_{st}(\theta_{s})\|^{3}=O_{p}(1),$$
where
$\widehat{g}_{st}(\theta_{s})=g_{s}(Y_{st},\widehat{I}_{st-1},\theta_{s})$, and
$\widehat{I}_{st}$ is defined as in (\ref{residual}).\label{asm2.4}
\end{ass}

\begin{ass}
The kernel functions $k$ and $l$ are symmetric, and both of them
and their partial derivatives up to second order are all uniformly bounded and Lipschitz continuous,
that is,
\begin{flalign*}
(i)\,\,\sup_{x,y}\|p(x,y)\|\leq C;
(ii)\,\,\|p(x_{1},y_{1})-p(x_{2},y_{2})\|\leq C\|(x_{1},y_{1})-(x_{2},y_{2})\|,
\end{flalign*}
for $p=k,k_{x},k_{y},k_{xx},k_{xy},k_{yy},
l,l_{x},l_{y},l_{xx},l_{xy},l_{xy}$, where $k_{x}=\p_{x} k(x,y)$, $k_{xy}=\p_{x}\p_{y} k(x,y)$, $l_{x}=\p_{x} l(x,y)$, and $l_{xy}=\p_{x}\p_{y} l(x,y)$.
\label{asm2.5}
\end{ass}

We offer some remarks on the above assumptions. Assumption
\ref{asm2.1} is standard for time series models. Assumption
\ref{asm2.2}(i) requires some technical moment conditions for the partial
derivatives of $g_{st}$. Assumption
\ref{asm2.2}(ii) presents some temporal dependence condition on the joint sequence $\{(\eta_{1t}^{T},\eta_{2t}^{T})^{T}\}$.
Assumption \ref{asm2.3} is satisfied under mild conditions
for most estimators,  such as (quasi) maximum likelihood estimator (MLE), least squares estimator (LSE), nonlinear least squares estimator (NLSE) and their robust modifications; see, e.g., Comte and Lieberman (2003), L\"{u}tkepohl (2005),  and Hafner and Preminger (2009) for more details. Assumption
\ref{asm2.4} is a condition on the truncation of the information set
$\widehat{I}_{st-1}$ and is similar to
Assumption A5 in Escanciano (2006). Assumption \ref{asm2.5} gives
some restrictive conditions for kernel functions $k$ and $l$; these conditions may exclude some kernel functions such as the fractional Brownian motion kernel, but they are usually satisfied by the often used Gaussian kernel, Laplace kernel and inverse multi-quadratics kernel.  The conditions in Assumptions \ref{asm2.1}-\ref{asm2.5} may be further relaxed, but they are convenient for
presenting our proofs in a simple way.

\subsection{Some lemmas}

This subsection gives some useful lemmas, which are key to study the asymptotics of our test statistics.

Before introducing these lemmas, we present some notations. Let
\begin{flalign}
\overline{k}_{ij}&=\frac{\p
g_{1i}(\theta_{10})}{\p\theta_{1}}k_{x}(\eta_{1i},\eta_{1j})+
\frac{\p
g_{1j}(\theta_{10})}{\p\theta_{1}}k_{y}(\eta_{1i},\eta_{1j}),\label{3.1}\\
\overline{l}_{qr}&=\frac{\p
g_{2q+m}(\theta_{20})}{\p\theta_{2}}l_{x}(\eta_{2q+m},\eta_{2r+m})+
\frac{\p
g_{2r+m}(\theta_{20})}{\p\theta_{2}}l_{y}(\eta_{2q+m},\eta_{2r+m}),\label{3.2}\\
\widecheck{k}_{ij}&= \left(\frac{\p
g_{1i}(\theta_{10})}{\p\theta_{1}}, \frac{\p
g_{1j}(\theta_{10})}{\p\theta_{1}}\right)\left(
\begin{array}{cc}
k_{xx}(\eta_{1i},\eta_{1j}) & k_{xy}(\eta_{1i},\eta_{1j})\\
k_{xy}(\eta_{1i},\eta_{1j}) & k_{yy}(\eta_{1i},\eta_{1j})
\end{array}
\right)\nonumber\\
&\quad\quad\times\left(\frac{\p g_{1i}(\theta_{10})}{\p\theta_{1}},
\frac{\p g_{1j}(\theta_{10})}{\p\theta_{1}}\right)^{T},\label{3.3}\\
\widecheck{l}_{qr}&= \left(\frac{\p
g_{2q+m}(\theta_{20})}{\p\theta_{2}}, \frac{\p
g_{2r+m}(\theta_{20})}{\p\theta_{2}}\right)\nonumber\\
&\quad\quad\times\left(
\begin{array}{cc}
l_{xx}(\eta_{2q+m},\eta_{2r+m}) & l_{xy}(\eta_{2q+m},\eta_{2r+m})\\
l_{xy}(\eta_{2q+m},\eta_{2r+m}) & l_{yy}(\eta_{2q+m},\eta_{2r+m})
\end{array}
\right)\nonumber\\
&\quad\quad\quad\quad\times\left(\frac{\p g_{2q+m}(\theta_{20})}{\p\theta_{2}},
\frac{\p
g_{2r+m}(\theta_{20})}{\p\theta_{2}}\right)^{T}\label{3.4}
\end{flalign}
for $i, j, q, r\in\{1, 2, \cdots, N\}$. With these notations, define
\begin{flalign}
S_{1n}^{(0)}(m)&=\frac{1}{N^{2}}\sum_{i,j}k_{ij}l_{ij}+\frac{1}{N^{4}}\sum_{i,j,q,r}k_{ij}l_{qr}
-\frac{2}{N^{3}}\sum_{i,j,q}k_{ij}l_{iq}, \label{3.5}\\
S_{1n}^{(ab)}(m)&=\frac{1}{N^{2}}\sum_{i,j}k_{ij}^{(ab)}l_{ij}^{(ab)}+\frac{1}{N^{4}}\sum_{i,j,q,r}k_{ij}^{(ab)}l_{qr}^{(ab)}
-\frac{2}{N^{3}}\sum_{i,j,q}k_{ij}^{(ab)}l_{iq}^{(ab)} \label{3.6}
\end{flalign}
for $a\in\{1,2\}$ and $b\in\{1,\cdots,a+1\}$, where
$k_{ij}^{(11)}=\overline{k}_{ij}$, $l_{ij}^{(11)}=l_{ij}$,
$k_{ij}^{(12)}=k_{ij}$, $l_{ij}^{(12)}=\overline{l}_{ij}$,
$k_{ij}^{(21)}=\widecheck{k}_{ij}$, $l_{ij}^{(21)}=l_{ij}$,
$k_{ij}^{(22)}=k_{ij}$, $l_{ij}^{(22)}=\widecheck{l}_{ij}$,
$k_{ij}^{(23)}=\overline{k}_{ij}$, and
$l_{ij}^{(23)}=\overline{l}_{ij}^{T}$.
Then,
$S_{1n}^{(0)}(m)$ can be expressed as the $V$-statistic of the form (see Gretton et al. 2005):
\begin{flalign*}
S_{1n}^{(0)}(m)=\frac{1}{N^{4}}\sum_{i,j,q,r}h^{(0)}_{m}(\eta_{i}^{(m)},\eta_{j}^{(m)},\eta_{q}^{(m)},\eta_{r}^{(m)})
\end{flalign*}
for some symmetric kernel $h^{(0)}_{m}$ given by
\begin{flalign*}
h^{(0)}_{m}(\eta_{i}^{(m)},\eta_{j}^{(m)},\eta_{q}^{(m)},\eta_{r}^{(m)})=
\frac{1}{4!}\sum_{(t,u,v,w)}^{(i,j,q,r)}\left(k_{tu}l_{tu}+k_{tu}l_{vw}-2k_{tu}l_{tv}\right),
\end{flalign*}
where the sum is taken over all $4!$ permutations of $(i,j,q,r)$, and $\eta_{t}^{(m)}=(\eta_{1t},\eta_{2t+m})\in\mathcal{R}^{d_{1}}\times\mathcal{R}^{d_{2}}$.
Likewise, all $S_{1n}^{(ab)}(m)$ can be expressed as the $V$-statistics
for the symmetric kernel $h^{(ab)}_{m}$ given by
\begin{flalign*}
h^{(ab)}_{m}(\varsigma_{i}^{(m)},\varsigma_{j}^{(m)},\varsigma_{q}^{(m)},\varsigma_{r}^{(m)})
=\frac{1}{4!}\sum_{(t,u,v,w)}^{(i,j,q,r)}\left(k_{tu}^{(ab)}l_{tu}^{(ab)}+
k_{tu}^{(ab)}l_{vw}^{(ab)}-2k_{tu}^{(ab)}l_{tv}^{(ab)}\right),
\end{flalign*}
where the sum is taken over all $4!$ permutations of $(i,j,q,r)$, and
$$\varsigma_{t}^{(m)}=\left(\eta_{1t},\frac{\p
g_{1t}(\theta_{10})}{\p\theta_{1}},\eta_{2t+m},\frac{\p g_{2t+m}(\theta_{20})}{\p\theta_{2}}\right)
\in \mathcal{R}^{d_{1}}\times
\mathcal{R}^{p_{1}\times d_{1}}\times \mathcal{R}^{d_{2}}\times \mathcal{R}^{p_{2}\times d_{2}}.$$

Now, we are ready to introduce these three lemmas.
The first lemma below gives an important expansion of $S_{1n}(m)$.

\begin{lem} \label{lem3.1}
$S_{1n}(m)$ admits the following expansion:
\begin{flalign*} 
S_{1n}(m)&=S_{1n}^{(0)}(m)+\zeta_{1n}^{T}S_{1n}^{(11)}(m)+\zeta_{2n}^{T}S_{1n}^{(12)}(m)\nonumber\\
&\quad+
\frac{1}{2}\zeta_{1n}^{T}S_{1n}^{(21)}(m)\zeta_{1n}+\frac{1}{2}\zeta_{2n}^{T}S_{1n}^{(22)}(m)\zeta_{2n}+
\zeta_{1n}^{T}S_{1n}^{(23)}(m)\zeta_{2n}+R_{1n}(m),
\end{flalign*}
where $S_{1n}^{(0)}(m)$ and $S_{1n}^{(ab)}(m)$ are defined as in (\ref{3.5}) and (\ref{3.6}), respectively,
$R_{1n}(m)$ is the remainder term,
and $\zeta_{sn}=\widehat{\theta}_{sn}-\theta_{s0}$.
\end{lem}

The second lemma below is crucial in studying the asymptotics of
$S_{1n}^{(0)}(m)$ and $S_{1n}^{(ab)}(m)$ under $H_{0}$.

\begin{lem}\label{lem3.2}
Suppose Assumptions \ref{asm2.1}, \ref{asm2.2}(i) and \ref{asm2.5} hold. Then, under $H_{0}$,
\begin{flalign*}
(i)\,\, &E[h^{(0)}_{m}(x_{1},\eta_{2}^{(m)},\eta_{3}^{(m)},\eta_{4}^{(m)})]=0
\end{flalign*}
for all $x_{1}\in\mathcal{R}^{d_{1}}\times\mathcal{R}^{d_{2}}$;
\begin{flalign*}
(ii)\,\,&E[h^{(ab)}_{m}(x_{1},\varsigma_{2}^{(m)},\varsigma_{3}^{(m)},\varsigma_{4}^{(m)})]=0
\end{flalign*}
for all $x_{1}\in
\mathcal{R}^{d_{1}}\times \mathcal{R}^{p_{1}\times d_{1}}\times
\mathcal{R}^{d_{2}}\times \mathcal{R}^{p_{2}\times d_{2}}$ and each $a, b=1,2$;
\begin{flalign*}
(iii)\,\,&E[h^{(23)}_{m}(x_{1},\varsigma_{2}^{(m)},\varsigma_{3}^{(m)},\varsigma_{4}^{(m)})]=\Upsilon
\end{flalign*}
for all $x_{1}\in
\mathcal{R}^{d_{1}}\times \mathcal{R}^{p_{1}\times d_{1}}\times
\mathcal{R}^{d_{2}}\times \mathcal{R}^{p_{2}\times d_{2}}$, where
\begin{flalign*}
\Upsilon=&4E\left[\frac{\p g_{12}(\theta_{10})}{\p\theta_{1}}k_{x}(\eta_{12},\eta_{11})\right]
E\left[\frac{\p g_{22}(\theta_{20})}{\p\theta_{2}}l_{x}(\eta_{22},\eta_{21})-\frac{\p g_{23}(\theta_{20})}{\p\theta_{2}}l_{x}(\eta_{23},\eta_{21})\right]\\
&+4E\left[\frac{\p g_{13}(\theta_{10})}{\p\theta_{1}}k_{x}(\eta_{13},\eta_{11})\right]
E\left[\frac{\p g_{23}(\theta_{20})}{\p\theta_{2}}l_{x}(\eta_{23},\eta_{21})-\frac{\p g_{22}(\theta_{20})}{\p\theta_{2}}l_{x}(\eta_{22},\eta_{21})\right].
\end{flalign*}
\end{lem}

By standard arguments for V-statistics (see, e.g., Lee (1990)), we have
$N[S_{1n}^{(0)}(m)]=N[V_{1n}^{(0)}(m)]+o_{p}(1)$, where
\begin{flalign}\label{3.7}
V_{1n}^{(0)}(m)=\frac{1}{N^{2}}\sum_{i,j}h_{2m}^{(0)}(\eta_{i}^{(m)},\eta_{j}^{(m)})
\end{flalign}
is the $V$-statistic with the kernel function
\begin{flalign}\label{h2m0}
h_{2m}^{(0)}(x_{1},x_{2})=E\left[h^{(0)}_{m}(x_{1},x_{2},\eta_{3}^{(m)},\eta_{4}^{(m)})\right]
\end{flalign}
for $x_{1}, x_{2}\in\mathcal{R}^{d_{1}}
\times\mathcal{R}^{d_{2}}$. Under $H_{0}$, $\{\eta_{t}^{(m)}\}$ is a sequence of i.i.d. random variables,
and hence Lemma \ref{lem3.2}(i) implies that $V_{1n}^{(0)}(m)$
is a degenerate $V$-statistic of order 1, from which
$h_{2m}^{(0)}$ can be expressed as
\begin{flalign}\label{3.8}
h_{2m}^{(0)}(x_{1},x_{2})=\sum_{j=0}^{\infty}\lambda_{jm}\Phi_{jm}(x_{1})\Phi_{jm}(x_{2}),
\end{flalign}
where $\{\Phi_{jm}(\cdot)\}$ is an orthonormal function in $L_{2}$ norm, and
$\lambda_{jm}$ is the eigenvalue corresponding to the eigenfunction $\Phi_{jm}(\cdot)$. That is,
 $\{\lambda_{jm}\}$ is a
finite enumeration of the nonzero eigenvalues of the
equation
$$E[h_{2m}^{(0)}(x_{1},\eta_{1}^{(m)})\Phi_{jm}(\eta_{1}^{(m)})]=\lambda_{jm}\Phi_{jm}(x_{1}),$$ where
 $E\Phi_{jm}(\eta_{1}^{(m)})=0$ for all $j\geq1$, and
$$E[\Phi_{jm}(\eta_{1}^{(m)})\Phi_{j'm}(\eta_{1}^{(m)})]=
\left\{\begin{array}{ll}
1 & j=j',\\
0 & j\not=j'
\end{array}\right.$$
(see, e.g.,  Dunford and Schwartz (1963, p.1087)).
With (\ref{3.7}) and (\ref{3.8}), we can obtain that under $H_{0}$,
\begin{flalign}\label{3.9}
N[S_{1n}^{(0)}(m)]=\sum_{j=1}^{\infty}\lambda_{jm}\left[\frac{1}{\sqrt{N}}\sum_{i=1}^{N}\Phi_{jm}(\eta_{i}^{(m)})\right]^{2}+o_{p}(1).
\end{flalign}

Moreover, we consider $S_{1n}^{(ab)}(m)$, which results from the estimation effect. Under $H_{0}$,
$S_{1n}^{(ab)}(m)$ (for $a, b=1, 2$) is a degenerate $V$-statistic of order 1 by Lemma \ref{lem3.2}(ii), and hence $N[S_{1n}^{(ab)}(m)]=O_{p}(1)$,
and then its related estimation effect is negligible in view of that
$\zeta_{sn}^{T}N[S_{1n}^{(ab)}(m)]=o_{p}(1)$. However, under $H_{0}$, the estimation effect related to
$S_{1n}^{(23)}(m)$ is negligible only when $\Upsilon=0$. This is because
when $\Upsilon\not=0$,
$S_{1n}^{(23)}(m)=O_{p}(1)$ by the law of large numbers for V-statistics, and
its related estimation effect is not negligible based on the ground that $N[\zeta_{1n}^{T}S_{1n}^{(23)}(m)\zeta_{2n}]=O_{p}(1)$.

Our third lemma below provides a useful central limit theorem.

\begin{lem}\label{lem3.3}
Suppose Assumptions \ref{asm2.1}, \ref{asm2.2}(i) and \ref{asm2.3}-\ref{asm2.5} hold. Then, under $H_{0}$,
\begin{flalign*}
\mathcal{T}_{n}&:=\left(\frac{1}{\sqrt{N}}\sum_{i=1}^{N}\mathcal{T}_{1i}^{T}, \frac{1}{\sqrt{n}}\sum_{i=1}^{n}
\mathcal{T}_{2i}^{T}\right)^{T}\to_{d} \mathcal{T}:=
((\mathcal{Z}_{jm})_{j\geq1,0\leq m\leq M},(\mathcal{W}_{s}^{T})_{1\leq s\leq 2})^{T}
\end{flalign*}
as $n\to\infty$, where
$\mathcal{T}_{1i}=\left(\big(\Phi_{jm}(\eta_{i}^{(m)})\big)_{j\geq1, 0\leq m\leq M}\right)^{T}$, $\mathcal{T}_{2i}=\left(\pi_{1i}^{T},\pi_{2i}^{T}\right)^{T}$,
$\mathcal{T}$ is a multivariate normal distribution with mean zero and covariance
matrix $\overline{\mathcal{T}}=E(\mathcal{T}_{1}\mathcal{T}_{1}^{T})$ with $\mathcal{T}_{i}=(\mathcal{T}_{1i}^{T}, \mathcal{T}_{2i}^{T})^{T}$, $\{\mathcal{Z}_{jm}\}_{j\geq1}$ is a
sequence of i.i.d. $\mbox{N}(0,1)$ random variables, and
$\mathcal{W}_{s}$ is a $p_{s}$-variate normal
random variable.
\end{lem}

\subsection{Asymptotics of test statistics}

Based on Lemmas \ref{lem3.1}-\ref{lem3.2}, this subsection studies the asymptotics of our test statistics. Let
\begin{flalign}\label{lambda_m23}
\Lambda_{m}^{(23)}:=E[h^{(23)}_{m}(\varsigma_{1}^{(m)},\varsigma_{2}^{(m)},\varsigma_{3}^{(m)},\varsigma_{4}^{(m)})].
\end{flalign}
First, we give the limiting null distributions of $S_{1n}(m)$ and $J_{1n}(M)$ as follows.

\begin{thm}\label{thm3.1}
Suppose Assumptions \ref{asm2.1}, \ref{asm2.2}(i) and \ref{asm2.3}-\ref{asm2.5} hold. Then, under $H_{0}$,
\begin{flalign*}
(i)\,\, & n[S_{1n}(m)]\to_{d}\chi_{m}\,\,\,\mbox{ for }0\leq m\leq M;\\
(ii)\,\, & n[J_{1n}(M)]\to_{d}\sum_{m=0}^{M}\chi_{m},
\end{flalign*}
as $n\to\infty$, where $\chi_{m}$ is a Gaussian process defined by
\begin{flalign*}
\chi_{m}=\sum_{j=1}^{\infty}\lambda_{jm}\mathcal{Z}_{jm}^{2}+
\mathcal{W}_{1}^{T}
\Lambda_{m}^{(23)}\mathcal{W}_{2}.
\end{flalign*}
Here,
$\lambda_{jm}$ is defined as in (\ref{3.8}), and $\mathcal{Z}_{jm}$ and $\mathcal{W}_{s}$ are defined as in Lemma \ref{lem3.3}.
\end{thm}

Theorem \ref{thm3.1} shows that $S_{1n}(m)$ and $J_{1n}(M)$ have convergence rate $n^{-1}$ under $H_{0}$.
Based on this theorem, we reject $H_{0}$
at the significance level $\alpha$, if
$$n[S_{1n}(m)]>c_{m\alpha}\,\,\,\mbox{ or }\,\,\,n[J_{1n}(M)]>c_{\alpha},$$
where $c_{m\alpha}$ and $c_{\alpha}$ are the
$\alpha$-th upper percentiles of $\chi_{m}$ and $\sum_{m=0}^{M}\chi_{m}$, respectively.
Since the distribution of
$\chi_{m}$ depends on $\{Y_{st}\}$ and $\{\pi_{st}\}$,
a residual bootstrap method is
proposed in Section 4 to obtain the values of $c_{m\alpha}$ and $c_{\alpha}$.

Second, we study the behavior of $S_{1n}(m)$ under the
following fixed alternative:
\begin{flalign*}
H_{1}^{(m)}: \,\,&\{\eta_{1t}\} \mbox{ and } \{\eta_{2t}\} \mbox{ are dependent such that }E[h_{2m}^{(0)}(x_{1},\eta_{2}^{(m)})]\not=0 \\
&\mbox{ for some }x_{1}\in\mathcal{R}^{d_{1}}\times\mathcal{R}^{d_{2}}.
\end{flalign*}
Under $H_{1}^{(m)}$, $h_{2m}^{(0)}$ is not a degenerate kernel of order 1. Hence, the V-statistic
$S_{1n}^{(0)}(m)$ can not have the convergence rate $n^{-1}$ as suggested by Lemma \ref{lem3.2}(i), leading to
the consistency of $S_{1n}(m)$ in detecting $H_{1}^{(m)}$. Similarly,
we can show the consistency of $J_{1n}(M)$ to detect the fixed alternative below:
\begin{flalign*}
H_{1}^{(M)}: H_{1}^{(m)} \mbox{ holds for some }m\in\{0, 1, \cdots, M\}.
\end{flalign*}

\begin{thm} \label{thm3.2}
Suppose Assumptions \ref{asm2.1}-\ref{asm2.5} hold. Then,
\begin{flalign*}
(i)\,\, &\lim_{n\to\infty}P\left(n[S_{1n}(m)]>c_{m\alpha}\right)=1\mbox{ under }H_{1}^{(m)};\\
(ii)\,\, &\lim_{n\to\infty}P\left(n[J_{1n}(M)]>c_{\alpha}\right)=1\mbox{ under }H_{1}^{(M)}.
\end{flalign*}
\end{thm}

In the end, we highlight that similar results as in Theorems \ref{thm3.1}-\ref{thm3.2} hold for
$S_{2n}(m)$ and $J_{2n}(M)$, which can be implemented in the similar way as $S_{1n}(m)$ and $J_{1n}(M)$, respectively.



\section{Residual bootstrap approximations}
In this section, we introduce a residual bootstrap method to
approximate the limiting null distributions in Theorem \ref{thm3.1}.
The residual bootstrap method has been well used in the time series
literature; see, e.g., Berkowitz and Kilian (2000), Paparoditis and
Politis (2003), Politis (2003), and many
others. Our residual bootstrap procedure to obtain the approximation
of the critical values $c_{m\alpha}$ and $c_{\alpha}$ is as follows:

{\it Step 1}. Estimate the original model (\ref{2.8}) and obtain
the residuals $\{\widehat{\eta}_{st}\}_{t=1}^{n}$.

{\it Step 2}. Generate bootstrap innovations $\{\widehat{\eta}_{st}^{*}\}_{t=1}^{n}$ (after standardization) by resampling with replacement from
the empirical residuals $\{\widehat{\eta}_{st}\}_{t=1}^{n}$.

{\it Step 3}. Given $\widehat{\theta}_{sn}$ and
$\{\widehat{\eta}_{st}^{*}\}_{t=1}^{n}$, generate bootstrap data set
$\{Y_{st}^{*}\}_{t=1}^{n}$ according to
$$Y_{st}^{*}=f_{s}(\widehat{I}_{st-1}^{*},\widehat{\theta}_{sn},\widehat{\eta}_{st}^{*}),$$
where $\widehat{I}_{st}^{*}$ is the bootstrap observable information set up to time $t$, conditional on some assumed initial values.

{\it Step 4}. Based on $\{Y_{st}^{*}\}_{t=1}^{n}$, compute
$\widehat{\theta}_{sn}^{*}$ in the same way as for
$\widehat{\theta}_{sn}$, and then calculate the corresponding bootstrap
residuals $\{\widehat{\eta}_{st}^{**}\}_{t=1}^{n}$ with
$\widehat{\eta}_{st}^{**}:=g_{s}(Y_{st}^{*},\widehat{I}_{st-1}^{*},\widehat{\theta}_{sn}^{*})$.

{\it Step 5}. Calculate the bootstrap test statistic $S_{1n}^{**}(m)$ and $J_{1n}^{**}(M)$ in
the same way as for (\ref{2.10}) and (\ref{2.12}), respectively, with $\widehat{\eta}_{st}^{**}$
replacing $\widehat{\eta}_{st}$.

{\it Step 6}. Repeat steps 1-5 $B$ times to obtain
$\{n[S_{1nb}^{**}(m)]; b=1, 2, \cdots, B\}$ and $\{n[J_{1nb}^{**}(M)]; b=1, 2, \cdots, B\}$, then choose their $\alpha$-th
upper percentiles, denoted by
$c_{\alpha}^{*}$ and $c_{m\alpha}^{*}$, as the approximations of $c_{\alpha}$ and $c_{m\alpha}$, respectively.

In order to prove the validity of the bootstrap procedure in steps 1-6, we need some notations.
Let
\begin{flalign}
h_{2m}^{(0*)}(x_{1},x_{2})&=E^{*}[h^{(0)}_{m}(x_{1},x_{2},\widehat{\eta}_{3}^{(m*)},\widehat{\eta}_{4}^{(m*)})], \label{h2m0_star}\\
\Lambda^{(23*)}_{m}&=E^*[h^{(23)}_{m}(\widehat{\varsigma}_{1}^{(m*)},\widehat{\varsigma}_{2}^{(m*)},
\widehat{\varsigma}_{3}^{(m*)},\widehat{\varsigma}_{4}^{(m*)})],\label{lambda_m23_star}
\end{flalign}
where $\widehat{\eta}_{t}^{(m*)}=(\widehat{\eta}_{1t}^{*}, \widehat{\eta}_{2t+m}^{*})$ and
$
{\varsigma}_{t}^{(m*)}=\big(\widehat{\eta}_{1t}^{*},\frac{\p
g_{1t}(\widehat{\theta}_{1n})}{\p\theta_{1}},\widehat{\eta}_{2t+m}^{*},\frac{\p g_{2t+m}(\widehat{\theta}_{2n})}{\p\theta_{2}}\big).
$
Also, let $\zeta_{sn}^*=\widehat{\theta}_{sn}^{*}-\widehat{\theta}_{sn}$,
and $\varpi_{n}:=\{Y_{11},Y_{12},\cdots,Y_{1n},$ $Y_{21}, Y_{22},\cdots,Y_{2n}\}$
be the given sample. Denote by $E^{*}$
the expectation conditional on $\varpi_{n}$; by $o_{p}^{*}(1) (O_{p}^{*}(1))$
a sequence of random variables converging to zero (bounded) in probability conditional on $\varpi_{n}$.

Since $\{\widehat{\eta}_{st}^*\}_{t=1}^N$ is an i.i.d sequence conditional on $\varpi_{n}$, a similar argument as for
Lemma \ref{lem3.1} implies that
\begin{flalign}\label{4.1}
S_{1n}^{**}(m)&=S_{1n}^{(0*)}(m)+\zeta_{1n}^{*T}S_{1n}^{(11*)}(m)+\zeta_{2n}^{*T}S_{1n}^{(12*)}(m) \nonumber\\
&\quad+\frac{1}{2}\zeta_{1n}^{*T}S_{1n}^{(21*)}(m)\zeta_{1n}^*
+\frac{1}{2}\zeta_{2n}^{*T}S_{1n}^{(22*)}(m)\zeta_{2n}^*+\zeta_{1n}^{*T}S_{1n}^{(23*)}(m)\zeta_{2n}^*+R_{1n}^*(m),
\end{flalign}
where $S_{1n}^{(0*)}(m)$, $S_{1n}^{(ab*)}(m)$ and $R_{1n}^*(m)$ are defined in the same way as $S_{1n}^{(0)}(m)$, $S_{1n}^{(ab)}(m)$ and $R_{1n}(m)$, respectively, with $\eta_{t}^{(m)}$ and $\varsigma_{t}^{(m)}$ being replaced by $\widehat{\eta}_{t}^{(m*)}$ and
$\widehat{\varsigma}_{t}^{(m*)}$, respectively.
Moreover, by a similar argument as for Lemma \ref{lem3.1}(i), we can obtain
\begin{flalign}
N[S_{1n}^{(0*)}(m)]&=\sum_{j=1}^{\infty}\lambda_{jm}^*\left[\frac{1}{\sqrt{N}}\sum_{i=1}^N\Phi_{jm}^*(\widehat{\eta}_i^{(m*)})\right]
+o_{p}^{*}(1),\label{4.2}
\end{flalign}
where
$E^{*}\Phi_{jm}^{*}(\widehat{\eta}_{1}^{(m*)})=0$ for all $j\geq1$, and
$E^{*}[\Phi_{jm}^{*}(\widehat{\eta}_{1}^{(m*)})\Phi_{j'm}^{*}(\widehat{\eta}_{1}^{(m*)})]=1$ if $j=j'$, and 0 if
$j\not=j'$.

Next, we give two technical assumptions.
\begin{ass}\label{asm4.1}
The bootstrap estimator $\widehat{\theta}_{sn}^*$  satisfies that
\begin{flalign*}
\sqrt{n}(\widehat{\theta}_{sn}^*-\widehat{\theta}_{sn})
&=\frac{1}{\sqrt{n}}\sum\limits_{t=1}^n\pi_s(Y_{st}^*,\widehat{I}_{st-1},\widehat{\theta}_{sn})+o_{p}^{*}(1)\\
&=:\frac{1}{\sqrt{n}}\sum\limits_{t=1}^n\pi_{st}^*+o_{p}^{*}(1),
\end{flalign*}
where $\pi_s$ is defined as in Assumption \ref{asm2.3} and $E^*(\pi_{st}^*|\widehat{I}_{st-1}^*)=0$.
\end{ass}

\begin{ass}\label{asm4.2}
The following convergence results hold:
\begin{flalign*}
(i)\,\, & \frac{1}{n}\sum_{i=1}^nE^*[\pi_{si}^*\pi_{s'i}^{*T}]\rightarrow_p E\left[\pi_{s1} \pi_{s'1}^T\right];\\
(ii)\,\, &\frac{1}{N} \sum_{i=1}^NE^*[\Phi_{jm}^*(\widehat{\eta}_{i}^{(m*)})\pi_{si}^*]\rightarrow_p E[\Phi_{jm}({\eta}_1^{(m)})\pi_{s1}],
\end{flalign*}
as $n\to\infty$, for $s, s'=1,2$, $j\geq1$, and $m=0,1,\cdots,M$.
\end{ass}


\noindent
Assumptions \ref{asm4.1} and \ref{asm4.2} are standard to prove the validity of the bootstrap procedure, and they are similar to those in Assumption A7 of Escanciano (2006). For  the (quasi) MLE,
LSE and NLSE or, more generally, estimators resulting from a martingale estimating equation (see Heyde, 1997),
the function $\pi_s(\cdot)$ required in Assumption \ref{asm4.1} could be expressed as
$\pi_s(Y_{st},{I}_{st-1},{\theta}_{s})=\varrho_1(\eta_{st}(\theta_{s}))\times \varrho_2({I}_{st-1},\theta_{s})$
for some functions $\varrho_1(\cdot)$ and $\varrho_2(\cdot)$ with $E(\varrho_1(\eta_{st}(\theta_{s0})))=0$.
Then, in those cases, Assumptions \ref{asm4.1} and \ref{asm4.2} are satisfied under some mild conditions on the function $\varrho_2(\cdot)$.
Note that the calculation of the bootstrap estimator $\widehat{\theta}_{sn}^{*}$ in step 4 may be time-consuming for some times series models (e.g, multivariate ARCH-type models) when $n$ is large.
In view of Assumption \ref{asm4.1}, we suggest to generate $\widehat{\theta}_{sn}^{*}$ as
$$\widehat{\theta}_{sn}^{*}=\widehat{\theta}_{sn}+\frac{1}{n}\sum_{t}\pi_{s}(Y_{st}^{*},\widehat{I}_{st-1}^{*},\widehat{\theta}_{sn}).$$
This results in saving a lot of compute time. In Section 5, we will apply this method to
the conditional variance models, and find that it can
generate very precise critical values $c_{m\alpha}$ and $c_{\alpha}$ for the proposed HSIC-based tests.

The following theorem guarantees that when $B$
is large, our bootstrapped critical values $c_{m\alpha}$ and $c_{\alpha}$ from steps 1-6 are valid
under the null or the alternative hypothesis.

\begin{thm}\label{thm4.1}
Suppose Assumptions \ref{asm2.1}-\ref{asm2.5} and \ref{asm4.1}-\ref{asm4.2} hold. Then, conditional on $\varpi_{n}$,
(i) $n[S_{1n}^{**}(m)]=O_{p}^{*}(1)$ for $0\leq m\leq M$; (ii) $n[J_{1n}^{**}(M)]=O_{p}^{*}(1)$; moreover, under $H_{0}$,
\begin{flalign*}
(iii)\,\, & n[S_{1n}^{**}(m)]\to_{d}\chi_{m}\,\,\,\mbox{ for }0\leq m\leq M,\\
(iv)\,\, & n[J_{1n}^{**}(M)]\to_{d}\sum_{m=0}^{M}\chi_{m},
\end{flalign*}
in probability as $n\to\infty$, where $\chi_m$ is defined as in Theorem \ref{thm3.1}.
\end{thm}


\section{Simulation studies}

In this section, we compare the performance of our HSIC-based tests $S_{sn}(m)$ and $J_{sn}(M)$ $(s=1, 2 \mbox{ hereafter})$
with some well-known existing tests in finite samples.

\subsection{Conditional mean models}
We generate 1000 replications of sample size $n$ from the following two conditional mean models:
\begin{flalign}
&
\left\{\begin{array}{l}
Y_{1t}=\left(\begin{array}{cc}
0.4 & 0.1\\
-1 & 0.5
\end{array}
\right)Y_{1t-1}+\eta_{1t},\\
Y_{2t}=\left(\begin{array}{cc}
-1.5 & 1.2\\
-0.9 & 0.5
\end{array}
\right)Y_{2t-1}+\eta_{2t},\\
\end{array}
\right.
\label{5.1}
\end{flalign}
where $\{\eta_{1t}\}$ and $\{\eta_{2t}\}$ are two sequences of i.i.d. random vectors.
To generate $\{\eta_{1t}\}$ and $\{\eta_{2t}\}$, we need an
auxiliary sequence of i.i.d. multivariate normal random vectors $\{u_{t}\}$ with mean zero, where
$u_{t}=(u_{1t}, u_{2t}, u_{3t}', u_{4t}')'$ with $u_{1t}, u_{2t}\in\mathcal{R}$ and
$u_{3t}, u_{4t}\in\mathcal{R}^{2\times1}$, and its covariance matrix is given by
\begin{flalign*}
\Omega=\left(
\begin{array}{ccc}
\Omega_{1}   & 0_{2\times2} & 0_{2\times2}\\
0_{2\times2} & \Omega_{2}   & \Omega_{4}\\
0_{2\times2} & \Omega_{4}'  & \Omega_{3}
\end{array}
\right)
\end{flalign*}
with
\begin{flalign*}
\Omega_{\tau}=\left(
\begin{array}{cc}
1 & \rho_{\tau}\\
\rho_{\tau} & 1
\end{array}
\right)\mbox{ for }\tau=1,2,3,\mbox{ and }
\Omega_{4}=\left(
\begin{array}{cc}
\rho_{4} & \rho_{4}\\
\rho_{4} & \rho_{4}
\end{array}
\right).
\end{flalign*}
Here, we set $\rho_{2}=0.5$ and $\rho_{3}=0.75$ as in El Himdl and Roy (1997), which
have also considered model (\ref{5.1}) in their simulations.

Based on $\{u_{t}\}$, we consider six different error generating processes (EGPs):
\begin{flalign*}
 \mbox{EGP 1}:&\,\, \eta_{1t}=u_{3t}, \eta_{2t}=u_{4t} \mbox{ and } \rho_{4}=0;\\
 \mbox{EGP 2}:&\,\, \eta_{1t}=u_{3t}, \eta_{2t}=u_{4t} \mbox{ and } \rho_{4}=0.3;\\
 \mbox{EGP 3}:&\,\, \eta_{1t}=\f{u_{1t}^2+1}{\sqrt{6}}u_{3t},  \eta_{2t}=|u_{1t}|u_{4t} \mbox{ and } \rho_{4}=0;\\
 \mbox{EGP 4}:&\,\, \eta_{1t}=\f{u_{1t}^2+1}{\sqrt{6}}u_{3t},  \eta_{2t}=|u_{1t+3}|u_{4t} \mbox{ and } \rho_{4}=0;\\
 \mbox{EGP 5}:&\,\, \eta_{1t}=\f{u_{1t}^2+1}{\sqrt{6}}u_{3t},  \eta_{2t}=|u_{2t}|u_{4t}, \rho_{1}=0.8 \mbox{ and }\rho_{4}=0;\\
 \mbox{EGP 6}:&\,\, \eta_{1t}=u_{1t}u_{3t}, \eta_{2t}=u_{2t}u_{4t}, \rho_{1}=0.8 \mbox{ and }\rho_{4}=0.
\end{flalign*}

\noindent Clearly, each entry of $\eta_{1t}$ or $\eta_{2t}$ has mean zero and variance one. Let $\rho_{\eta_{1},\eta_{2}}(d)$ be the cross-correlation matrix between $\eta_{1t}$ and $\eta_{2t+d}$.
EGP 1 is designed for the null hypothesis, since $\rho_{\eta_{1},\eta_{2}}(d)=0_{2\times2}$ for all $d$ in this case.
EGPs 2-6 are set for the alternative hypotheses, since they pose a linear or non-linear dependence structure between
$\eta_{1t}$ and $\eta_{2t}$. Specifically, a linear dependence structure between $\eta_{1t}$ and $\eta_{2t}$ exists
in EGP 2, with $\rho_{\eta_{1},\eta_{2}}(d)=0.3I_{2}$ for $d=0$, and 0 otherwise; a non-linear
dependence structure between
$\eta_{1t}$ and $\eta_{2t}$ is induced by the co-factor $u_{1t}$ in EGP 3, the lagged co-factors $u_{1t}$ and $u_{1t+3}$ in EGP 4, and
two correlated co-factors $u_{1t}$ and $u_{2t}$ in EGPs 5 and 6. In EGPs 2-6,
$\eta_{1t}$ and $\eta_{2t}$ are dependent but un-correlated.

Now, we fit each replication by using the least squares estimation method for model
(\ref{5.1}). Denote by $\{\widehat{\eta}_{1t}\}$ and
$\{\widehat{\eta}_{2t}\}$ the residuals from the fitted models.
Based on $\{\widehat{\eta}_{1t}\}$ and
$\{\widehat{\eta}_{2t}\}$,
we compute
$S_{sn}(m)$ and $J_{sn}(M)$ ($S_{sn}$ and $J_{sn}$ in short), with $k$ and $l$ being the Gaussian kernels
and $\sigma=1$. The critical values of all HSIC-based tests are obtained by the
residual bootstrap method with $B=1000$ in Section 4.

Meanwhile, we also compute the test statistics $G_{sn}(M)$ ($G_{sn}$ in short) in El Himdl  and  Roy (1997) and the test statistics $W_{sn}(h)$ ($W_{sn}$ in short) in Bouhaddioui and Roy (2006),
where
\begin{flalign*}
&G_{1n}(M)=\sum_{m=-M}^{M}\widehat{Z}_{n}(m),\,\,\,G_{2n}(M)=\sum_{m=-M}^{M}[n/(n-|m|)]\widehat{Z}_{n}(m),\\
&W_{1n}(h)=\frac{\sum_{m=1-n}^{n-1}[\overline{K}(m/h)]^{2}\widetilde{Z}_{n}(m)-d_{1}d_{2}A_{1n}(h)}{\sqrt{2d_{1}d_{2}B_{1n}(h)}},\\
&W_{2n}(h)=\frac{\sum_{m=1-n}^{n-1}[\overline{K}(m/h)]^{2}\widetilde{Z}_{n}(m)-hd_{1}d_{2}A_{1}}{\sqrt{2hd_{1}d_{2}B_{1}}}.
\end{flalign*}
Here, $\widehat{Z}_{n}(m)=n[vec(R_{12}(m))]^{T}[R_{22}^{-1}(0)\otimes
R_{11}^{-1}(0)][vec(R_{12}(m))]$,
$R_{ij}(m)=D[(\widehat{r}_{ii}(0))^{-1/2}]$
$\widehat{r}_{ij}(m)D[(\widehat{r}_{jj}(0))^{-1/2}]$,
$\widehat{r}_{ij}(m)$ is the sample cross-covariance matrix between
$\{\widehat{\eta}_{it}\}$ and $\{\widehat{\eta}_{jt+m}\}$,
$\widetilde{Z}_{n}(m)$ is defined in the same way as $\widehat{Z}_{n}(m)$ with
$\widehat{\eta}_{st}$ being replaced by $\widetilde{\eta}_{st}$,
$\widetilde{\eta}_{st}$ is the residual from a fitted
VAR($p$) model for $Y_{st}$, $\overline{K}(\cdot)$ is a kernel function, $h$ stands for the bandwidth,
$A_{1}=\int_{-\infty}^{\infty}[\overline{K}(z)]^{2}dz$,
$B_{1}=\int_{-\infty}^{\infty}[\overline{K}(z)]^{4}dz$, and
\begin{flalign*}
A_{1n}(h)&=\sum_{m=1-n}^{n-1}(1-|m|/n)[\overline{K}(m/h)]^{2},\\
B_{1n}(h)&=\sum_{m=1-n}^{n-1}(1-|m|/n)(1-(|m|+1)/n)[\overline{K}(m/h)]^{4}.
\end{flalign*}
Note that $G_{1n}$ is for testing the cross-correlation between $\eta_{1t}$ and $\eta_{2t}$, and $G_{2n}$ is its modified version for small $n$; $W_{1n}$ is towards the same goal as $G_{1n}$ but with ability to detect the cross-correlation beyond lag $M$, and $W_{2n}$ is the modified version of $W_{1n}$.  Under certain
conditions, the limiting null distribution of $G_{1n}$ or $G_{2n}$
is $\chi^{2}_{(2M+1)d_{1}d_{2}}$, and that of $W_{1n}$ or $W_{2n}$
is $N(0,1)$.

In all simulation studies, we set $m=0$ and $3$ for the single HSIC-based tests $S_{sn}(m)$, and set $M=3$ and $6$ for the joint HSIC-based test $J_{sn}(M)$. Because $S_{1n}(0)=S_{2n}(0)$, the results of $S_{2n}(0)$ are absent.
For $G_{sn}(M)$, we choose $M=3, 6$ and $9$. For $W_{sn}(h)$, we follow Hong (1996) to
choose $p=3$ (or 6) when $n=100$ (or 200), and
use the kernel function $\overline{K}(z)=\sin(\pi
z)/(\pi z)$ (Daniel kernel) with the bandwidth $h=h_{1}, h_{2}$ or $h_{3}$, where  $h_{1}=[\log(n)], h_{2}=[3n^{0.2}]$, and $h_{3}=[3n^{0.3}]$. The
significance level $\alpha$ is set to be $1\%,5\%$ and 10\%.

\begin{table}
\caption{The sizes and power ($\times$100) of all tests for model (\ref{5.1}) at $\alpha=1\%,5\%$ and $10\%$}
\label{tab:1}       
\renewcommand{\arraystretch}{1.2}
\centering \scriptsize\addtolength{\tabcolsep}{-4.8pt}
\begin{tabular}{ c   ccc ccc   ccc ccc    ccc ccc  }
\hline
\multicolumn{1}{c}{}&\multicolumn{6}{c}{$\mbox{EGP 1}$}&\multicolumn{6}{c}{$\mbox{EGP 2}$}&\multicolumn{6}{c}{$\mbox{EGP 3}$}\\
\cmidrule(lr){2-7}\cmidrule(lr){8-13}\cmidrule(lr){14-19}
\multicolumn{1}{c}{}&\multicolumn{3}{c}{$n=100$}&\multicolumn{3}{c}{$n=200$}&\multicolumn{3}{c}{$n=100$}&\multicolumn{3}{c}{$n=200$}&\multicolumn{3}{c}{$n=100$}&\multicolumn{3}{c}{$n=200$}\\
\cmidrule(lr){2-7}\cmidrule(lr){8-13}\cmidrule(lr){14-19}
\multicolumn{1}{c}{Tests}&\multicolumn{1}{c}{1\%}&\multicolumn{1}{c}{5\%}&\multicolumn{1}{c}{10\%}&\multicolumn{1}{c}{1\%}&\multicolumn{1}{c}{5\%}&\multicolumn{1}{c}{10\%}&\multicolumn{1}{c}{1\%}&\multicolumn{1}{c}{5\%}&\multicolumn{1}{c}{10\%}&\multicolumn{1}{c}{1\%}&\multicolumn{1}{c}{5\%}&\multicolumn{1}{c}{10\%}&\multicolumn{1}{c}{1\%}&\multicolumn{1}{c}{5\%}&\multicolumn{1}{c}{10\%}&\multicolumn{1}{c}{1\%}&\multicolumn{1}{c}{5\%}&\multicolumn{1}{c}{10\%}\\
\cmidrule(lr){1-1}\cmidrule(lr){2-7}\cmidrule(lr){8-13}\cmidrule(lr){14-19}
$S_{1n}(0)$&0.7&5.1&11.7&1.6&5.2&11.7&47.1&69.1&79.9&85.5&95.2&97.4&80.2&94.5&97.9&99.3&100&100\\
$S_{1n}(3)$&0.6&5.4&11.4&0.7&4.3&10.9&1.1&5.5&13.0&0.6&4.9&9.9&0.8&5.1&10.6&1.1&5.9&10.0\\
$S_{2n}(3)$&1.2&5.6&12.1&1.3&4.6&9.9&1.0&5.1&11.4&1.5&5.3&9.9&1.0&5.5&11.2&0.8&4.1&9.1\\
$J_{1n}(3)$&0.7&5.3&12.3&1.2&5.2&11.5&19.4&44.5&58.4&55.1&78.4&85.4&30.7&64.4&79.9&88.0&96.8&98.8\\
$J_{1n}(6)$&0.9&6.2&14.6&1.1&6.1&13.6&12.5&32.4&48.2&40.3&66.1&76.8&11.6&37.0&55.7&66.4&89.0&95.1\\
$J_{2n}(3)$&1.4&7.1&12.5&1.8&6.7&13.9&19.3&42.2&57.4&54.8&78.3&87.0&31.9&61.7&77.6&86.7&96.8&98.3\\
$J_{2n}(6)$&1.1&6.8&13.2&1.7&6.5&12.1&13.2&32.9&47.3&38.3&62.7&76.6&10.4&36.9&56.0&66.0&87.5&94.1\\\\
$G_{1n}(3)$&0.5&3.6&7.6&0.7&5.0&10.1&17.3&41.5&57.1&69.1&88.4&93.0&10.9&23.9&33.4&14.7&29.3&39.4\\
$G_{1n}(6)$&0.4&2.8&7.8&0.6&4.2&9.6&17.3&41.5&57.1&43.5&70.9&83.5&5.3&14.6&24.9&8.5&21.6&32.8\\
$G_{1n}(9)$&0.4&1.5&4.9&0.2&3.3&6.8&8.1&25.0&39.1&29.4&55.1&69.3&2.9&10.0&16.6&6.3&17.0&25.2\\
$G_{2n}(3)$&0.9&4.2&8.6&0.7&5.5&10.5&18.3&43.3&59.4&69.5&89.0&93.6&11.9&25.2&35.5&15.2&29.9&40.7\\
$G_{2n}(6)$&0.6&4.6&10.4&1.0&5.4&10.9&12.5&30.3&45.0&45.8&72.8&84.4&6.6&18.4&29.6&10.2&24.4&34.8\\
$G_{2n}(9)$&0.7&4.1&9.1&0.6&4.5&9.5&7.9&25.4&36.6&34.1&60.2&74.7&5.0&15.7&23.8&8.3&19.9&28.8\\\\
$W_{1n}(h_{1})$&0.9&5.2&9.4&2.2&6.9&12.8&45.6&64.9&75.2&87.5&93.9&96.9&24.2&37.4&46.9&27.2&42.4&51.1\\
$W_{1n}(h_{2})$&0.8&4.3&8.4&1.7&6.3&12.4&30.3&53.0&65.7&78.3&89.4&93.4&18.8&30.3&39.4&21.4&36.9&46.0\\
$W_{1n}(h_{3})$&1.0&5.4&9.4&1.6&5.4&12.5&19.6&44.5&57.3&59.6&80.2&88.0&12.6&25.3&35.5&15.1&29.4&39.6\\
$W_{2n}(h_{1})$&0.6&4.2&7.6&2.1&6.2&11.7&41.1&62.4&72.9&86.1&93.2&96.5&21.6&35.6&44.3&25.7&40.9&50.0\\
$W_{2n}(h_{2})$&0.4&3.2&5.6&1.4&5.0&9.8&23.1&46.4&59.4&74.3&87.7&92.1&14.7&26.2&34.3&19.2&33.5&43.5\\
$W_{2n}(h_{3})$&0.3&1.7&4.9&0.9&3.3&6.8&11.0&28.5&43.3&49.5&73.8&83.0&8.2&17.9&24.9&10.3&22.8&31.7\\\\
\multicolumn{1}{c}{}&\multicolumn{6}{c}{$\mbox{EGP 4}$}&\multicolumn{6}{c}{$\mbox{EGP 5}$}&\multicolumn{6}{c}{$\mbox{EGP 6}$}\\
\cmidrule(lr){2-7}\cmidrule(lr){8-13}\cmidrule(lr){14-19}
\multicolumn{1}{c}{}&\multicolumn{3}{c}{$n=100$}&\multicolumn{3}{c}{$n=200$}&\multicolumn{3}{c}{$n=100$}&\multicolumn{3}{c}{$n=200$}&\multicolumn{3}{c}{$n=100$}&\multicolumn{3}{c}{$n=200$}\\
\cmidrule(lr){2-7}\cmidrule(lr){8-13}\cmidrule(lr){14-19}
\multicolumn{1}{c}{Tests}&\multicolumn{1}{c}{1\%}&\multicolumn{1}{c}{5\%}&\multicolumn{1}{c}{10\%}&\multicolumn{1}{c}{1\%}&\multicolumn{1}{c}{5\%}&\multicolumn{1}{c}{10\%}&\multicolumn{1}{c}{1\%}&\multicolumn{1}{c}{5\%}&\multicolumn{1}{c}{10\%}&\multicolumn{1}{c}{1\%}&\multicolumn{1}{c}{5\%}&\multicolumn{1}{c}{10\%}&\multicolumn{1}{c}{1\%}&\multicolumn{1}{c}{5\%}&\multicolumn{1}{c}{10\%}&\multicolumn{1}{c}{1\%}&\multicolumn{1}{c}{5\%}&\multicolumn{1}{c}{10\%}\\
\cmidrule(lr){1-1}\cmidrule(lr){2-7}\cmidrule(lr){8-13}\cmidrule(lr){14-19}
$S_{1n}(0)$&0.4&4.4&10.1&0.6&4.1&9.5&23.7&50.5&65.2&58.7&84.0&91.9&36.8&64.3&76.3&77.2&91.9&95.7\\
$S_{1n}(3)$ &0.4&3.7&7.9&0.4&3.9&9.5&0.5&4.2&9.2&0.7&4.3&9.4&0.5&3.1&7.8&0.8&4.7&9.8\\
$S_{2n}(3)$&75.5&92.0&96.3&99.2&99.9&100&0.7&1.0&3.5&3.0&4.5&9.1&0.4&3.0&7.6&0.6&4.4&8.4\\
$J_{1n}(3)$&0.3&2.6&6.5&0.4&2.7&7.8&4.5&23.6&34.4&20.7&46.3&60.4&7.6&25.3&41.9&35.8&63.8&75.7\\
$J_{1n}(6)$&0.3&1.7&5.2&0.2&2.1&5.3&1.3&9.5&19.3&9.0&28.8&45.4&1.7&12.4&25.5&17.9&40.5&57.5\\
$J_{2n}(3)$&28.4&57.2&76.2&86.7&96.5&98.5&4.7&21.5&32.4&19.3&45.7&59.7&5.6&23.6&38.8&35.4&63.0&75.9\\
$J_{2n}(6)$&9.7&34.3&53.7&64.4&88.1&94.6&1.9&8.5&19.4&8.8&27.5&45.9&1.8&10.3&23.4&11.3&22.9&31.9\\\\
$G_{1n}(3)$&10.4&21.4&31.9&12.8&27.1&38.4&5.5&14.7&23.7&8.1&19.6&28.0&3.9&12.7&20.3&4.9&14.2&24.8\\
$G_{1n}(6)$&4.6&13.7&21.4&8.4&19.8&30.2&2.0&9.6&16.7&3.9&14.2&24.6&2.8&8.8&15.2&2.9&10.6&16.3\\
$G_{1n}(9)$&2.9&8.3&15.4&5.4&15.6&24.5&1.4&5.3&12.3&2.7&10.6&17.5&1.7&6.9&11.2&2.1&7.9&13.9\\
$G_{2n}(3)$&12.3&24.7&35.5&13.8&28.6&39.7&6.1&15.9&25.3&8.3&20.2&29.4&4.2&13.7&22.9&5.0&14.6&25.5\\
$G_{2n}(6)$&7.0&17.8&26.8&9.0&22.9&32.6&3.2&12.8&21.3&4.6&16.5&26.1&3.7&11.6&19.3&3.3&11.6&19.0\\
$G_{2n}(9)$&4.8&14.6&25.8&7.0&19.6&27.9&2.6&11.1&19.5&4.5&13.0&22.5&3.1&10.4&18.7&2.7&9.8&17.6\\\\
$W_{1n}(h_{1})$&2.8&9.6&16.5&6.6&15.7&24.8&14.1&20.5&34.1&16.0&28.3&35.7&11.6&21.7&30.8&11.3&22.9&31.9\\
$W_{1n}(h_{2})$&7.9&16.9&25.1&10.9&23.6&34.1&10.5&19.2&29.4&12.9&23.5&34.2&8.1&17.4&27.0&8.8&18.3&27.6\\
$W_{1n}(h_{3})$&8.7&18.2&27.1&10.7&25.9&35.7&6.9&18.2&26.2&9.2&19.9&29.6&6.7&15.9&24.1&5.5&15.1&21.8\\
$W_{2n}(h_{1})$&2.3&8.2&14.1&6.3&14.8&23.4&13.2&19.9&32.1&15.5&26.9&34.2&10.0&19.7&22.6&10.5&21.9&30.2\\
$W_{2n}(h_{2})$&6.3&13.6&20.1&9.2&20.6&30.4&8.2&16.5&23.6&11.7&20.7&31.6&6.5&13.9&20.6&7.2&16.5&24.0\\
$W_{2n}(h_{3})$&5.6&11.8&17.5&8.3&18.2&29.1&4.0&10.8&17.5&6.5&15.1&21.3&3.2&9.3&15.4&3.6&10.3&16.9\\
\hline
\multicolumn{19}{l}{$\dagger\mbox{ For }W_{sn}, h_{1}=[\log(n)], h_{2}=[3n^{0.2}] \mbox{ and } h_{3}=[3n^{0.3}]$}
\end{tabular}
\end{table}

Table \ref{tab:1} reports the power of all tests for model
(\ref{5.1}), and the sizes of all tests are corresponding to those in EGP 1.
From this table, our findings are as follows:

(i) The sizes of all single HSIC-based tests $S_{sn}$  are close to
their nominal ones in most cases, while the sizes of other tests
are a little unsatisfactory. For instance, $J_{sn}$ are slightly oversized especially at $\alpha=5\%$ and $10\%$, while
 $W_{1n}$ (or $W_{2n}$) is slightly oversized (or undersized) when $n=200$ (or 100) at all levels. The size performance of
 $G_{sn}$ depends on $M$: a larger value of $M$ leads to a more undersized behavior especially at $\alpha=10\%$, although
 $G_{2n}$ in general has a better performance than $G_{1n}$.

(ii) In all examined cases, the single HSIC-based test $S_{1n}(0)$ is much more powerful than other tests in EGPs 2-3 and 5-6, and
the single HSIC-based test $S_{2n}(3)$ has a significant power advantage in EGP 4. These results are expected, since
$S_{1n}(0)$ and $S_{2n}(3)$ are tailored to examine the dependence at specific lags $0$ and $3$, respectively, which are
the set-ups of our EGPs.

(iii) For the linear dependence case (i.e., EGP 2), the joint HSIC-based tests $J_{sn}$ have a comparable power
performance as $G_{sn}$, and they are much less powerful than $W_{1n}(h_{1})$ but much more powerful than $W_{2n}(h_{3})$
when $n=100$. For the non-linear dependence case (i.e., EGPs 3-6), the joint HSIC-based tests $J_{sn}$ in general are much more powerful than
the tests $G_{sn}$ and $W_{sn}$ especially when $n=200$. The only exception is
$J_{1n}$ in EGP 4, since $J_{1n}$ can not detect the dependence between $\eta_{1t+m}$ and $\eta_{2t}$ at lag $m=3$.
In contrast, $J_{2n}$ performs very well here.

(iv) In all examined cases, the power of  $J_{sn}$ and $G_{sn}$ decreases as the value of $M$ increase, while
this tendency is vague for  $W_{sn}$.

Overall, our single HSIC-based tests are very powerful in detecting  dependence at  specific lags, and
our joint HSIC-based tests exhibit a significant power advantage in detecting  non-linear dependence, which can not be easily examined by other tests.

\subsection{Conditional variance models}
We generate 1000 replications of sample size $n$ from the following two conditional variance models:
\begin{flalign}
&\left\{\begin{array}{l}
Y_{1t}=V_{1t}^{1/2}\eta_{1t}\,\,\,\mbox{ and }\,\,\,V_{1t}=(v_{1t,ij})_{i,j=1,2},\\
Y_{2t}=V_{2t}^{1/2}\eta_{2t}\,\,\,\mbox{ and }\,\,\,V_{2t}=(v_{2t,ij})_{i,j=1,2},\\
\mbox{with }\\
\left(
\begin{array}{ccc}
v_{1t,11}\\
v_{1t,22}\\
v_{1t,12}\\
\end{array}
\right)=\left(
\begin{array}{cc}
0.2+0.5v_{1t-1,11}+0.1Y_{1 t-1,1}^2 \\
0.2+0.5v_{1t-1,22}+0.1Y_{1 t-1,2}^2 \\
0.5\sqrt{v_{1t-1,11}v_{1t-1,22}}
\end{array}
\right),\\\left(
\begin{array}{ccc}
v_{2t,11}\\
v_{2t,22}\\
v_{2t,12}\\
\end{array}
\right)=\left(
\begin{array}{cc}
0.3+0.4v_{2t-1,11}+0.2Y_{2t-1,1}^2 \\
0.3+0.4v_{2t-1,22}+0.2Y_{2t-1,2}^2 \\
0.6\sqrt{v_{2t-1,11}v_{2t-1,22}}
\end{array}
\right),\\
\end{array}
\right.\label{5.2}
\end{flalign}
where $\{\eta_{1t}\}$ and $\{\eta_{2t}\}$ are two sequences of i.i.d. random vectors generated as for model (\ref{5.1}).
Model (\ref{5.2}) contains two CC-MGARCH models studied in Tse (2002).
For each replication, we fit the above models by using the Gaussian-QMLE method. Denote by $\{\widehat{\eta}_{1t}\}$ and
$\{\widehat{\eta}_{2t}\}$ the residuals from the fitted models.
Based on $\{\widehat{\eta}_{1t}\}$ and
$\{\widehat{\eta}_{2t}\}$,
we compute
$S_{sn}(m)$ and $J_{sn}(M)$, and their critical values as for model (\ref{5.1}).

At the same time, we also compute the test statistics $L_{sn}(M)$ and $T_{sn}(M)$  ($L_{sn}$ and $T_{sn}$ in short) in Tchahou and Duchesne (2013), where
\begin{flalign*}
&L_{1n}(M)=\sum_{m=-M}^Mn\rho_{\widehat{q}_{1t},\widehat{q}_{2t}}^{2}(m),\ \  L_{2n}(M)=\sum_{m=-M}^M[n^2/(n-|m|)]\rho_{\widehat{q}_{1t},\widehat{q}_{2t}}^{2}(m),\\
&T_{1n}(M)=\sum_{m=-M}^Mn\cdot tr(C_{12}^{T}(m)C_{11}^{-1}(0)C_{12}(m)C_{22}^{-1}(0)),\\
&T_{2n}(M)=\sum_{m=-M}^M[n^{2}/(n-|m|)]\cdot tr(C_{12}^{T}(m)C_{11}^{-1}(0)C_{12}(m)C_{22}^{-1}(0)).
\end{flalign*}
Here, $\rho_{\widehat{q}_{1t},\widehat{q}_{2t}}(m)$ is the sample cross-correlation between $\{\widehat{q}_{1t}\}$  and $\{\widehat{q}_{2t+m}\}$,
$C_{ij}(m)$ is the sample cross-covariance matrix between $\{\widehat{\varphi}_{it}\}$ and $\{\widehat{\varphi}_{jt+m}\}$,
$\widehat{q}_{st}=\widehat{\eta}_{st}^{T}\widehat{\eta}_{st}$, and $\widehat{\varphi}_{st}=vech(\widehat{\eta}_{st}\widehat{\eta}_{st}^{T})$.
It is worth noting that $L_{1n}$ (or $T_{1n}$) is for testing the cross-correlation between two transformed (or original)
residuals, and $L_{2n}$ (or $T_{2n}$) is its modified version for small $n$.
Under certain
conditions, the limiting null distribution of $L_{1n}$ or $L_{2n}$
is $\chi^{2}_{(2M+1)}$, and that of $T_{1n}$ or $T_{2n}$
is $\chi^{2}_{(2M+1)d_{1}^{*}d_{2}^{*}}$, where $d_{s}^{*}=d_{s}(d_{s}+1)/2$ for $s=1, 2$.

In all simulation studies, we choose the values of $m$ and $M$ as for model (\ref{5.1}). The
significance level $\alpha$ is set to be $1\%,5\%$ and 10\%. Table \ref{tab:2} summarizes the
power results of all tests for model
(\ref{5.2}), and the sizes of all tests are corresponding to those in EGP 1.
From this table, our findings are as follows:

(i) The sizes of all tests are close to their nominal ones, although most of $T_{sn}$ are slightly oversized.

(ii) Similar to the results in model (\ref{5.1}), the single HSIC-based test $S_{1n}(0)$ or $S_{1n}(3)$ as expected is
the most powerful one among all tests.

(iii) For the linear dependence case (i.e., EGP 2), all joint HSIC-based tests $J_{sn}$ are much more powerful than
$L_{sn}$ and $T_{sn}$. For the non-linear dependence case (i.e., EGP 3-6), all $J_{sn}$ still have larger power than $L_{sn}$ and $T_{sn}$ in most cases, but this advantage is small especially for $J_{sn}(6)$. There are two exceptions that some $J_{sn}$ exhibit low power: first, $J_{1n}(3)$ and $J_{1n}(6)$ as argued for model (\ref{5.1}) have no power in EGP 4; second, $J_{2n}(6)$ is less powerful than most of $L_{sn}$ and $T_{sn}$ especially for $n=200$. Since the cross-correlation between $\eta_{1t}^2$ and $\eta_{2t}^2$ is high in EGPs 2-6, the relative good power performance of
$L_{sn}$ and $T_{sn}$ in some cases is not out of our expectation.

(iv) For the tests $J_{sn}$, $L_{sn}$ and $T_{sn}$, their power decreases as the value of $M$ increases in all examined cases.

\begin{table}
\begin{center}
\caption{The sizes and power ($\times$100) of all tests for model (\ref{5.2}) at $\alpha=1\%,5\%$ and $10\%$}
\label{tab:2}       
\renewcommand{\arraystretch}{1.2}
\centering \scriptsize\addtolength{\tabcolsep}{-4.8pt}
\begin{tabular}{ c   ccc ccc   ccc ccc    ccc ccc   }
\hline
\multicolumn{1}{c}{}&\multicolumn{6}{c}{$\mbox{EGP 1}$}&\multicolumn{6}{c}{$\mbox{EGP 2}$}&\multicolumn{6}{c}{$\mbox{EGP 3}$}\\
\cmidrule(lr){2-7}\cmidrule(lr){8-13}\cmidrule(lr){14-19}
\multicolumn{1}{c}{}&\multicolumn{3}{c}{$n=200$}&\multicolumn{3}{c}{$n=300$}&\multicolumn{3}{c}{$n=200$}&\multicolumn{3}{c}{$n=300$}&\multicolumn{3}{c}{$n=200$}&\multicolumn{3}{c}{$n=300$}\\
\cmidrule(lr){2-7}\cmidrule(lr){8-13}\cmidrule(lr){14-19}
\multicolumn{1}{c}{Tests}&\multicolumn{1}{c}{1\%}&\multicolumn{1}{c}{5\%}&\multicolumn{1}{c}{10\%}&\multicolumn{1}{c}{1\%}&\multicolumn{1}{c}{5\%}&\multicolumn{1}{c}{10\%}&\multicolumn{1}{c}{1\%}&\multicolumn{1}{c}{5\%}&\multicolumn{1}{c}{10\%}&\multicolumn{1}{c}{1\%}&\multicolumn{1}{c}{5\%}&\multicolumn{1}{c}{10\%}&\multicolumn{1}{c}{1\%}&\multicolumn{1}{c}{5\%}&\multicolumn{1}{c}{10\%}&\multicolumn{1}{c}{1\%}&\multicolumn{1}{c}{5\%}&\multicolumn{1}{c}{10\%}\\
\cmidrule(lr){1-1}\cmidrule(lr){2-7}\cmidrule(lr){8-13}\cmidrule(lr){14-19}
$S_{1n}(0)$&0.7&4.3&10.5&1.6&5.4&9.2&100&100&100&100&100&100&100&100&100&100&100&100\\
$S_{1n}(3)$&1.2&5.2&11.0&0.5&5.1&10.1&1.3&5.8&10.8&1.5&5.8&9.6&0.8&4.1&8.9&0.8&5.4&10.8\\
$S_{2n}(3)$&1.1&4.5&9.3&0.6&4.6&9.7&0.9&5.1&9.3&0.9&4.6&9.3&1.2&4.9&9.5&1.2&4.5&8.6\\
$J_{1n}(3)$&0.7&4.5&10.7&0.8&4.7&9.0&99.2&99.9&99.9&100&100&100&97.7&99.6&99.8&100&100&100\\
$J_{1n}(6)$&0.7&3.7&9.1&0.4&4.1&8.8&91.3&98.5&99.4&99.8&100&100&85.9&96.5&98.6&99.2&100&100\\
$J_{2n}(3)$&0.8&4.1&9.2&1.0&5.5&11.6&98.6&99.8&99.9&100&100&100&97.8&99.6&100&100&100&100\\
$J_{2n}(6)$&0.6&4.0&9.0&1.0&4.9&10.3&91.0&97.8&99.1&99.9&100&100&83.8&96.4&98.8&95.5&95.9&96.0\\\\

$L_{1n}(3)$&1.2&3.9&9.9&1.3&6.1&10.0&15.7&34.8&46.3&32.2&54.3&65.4&87.6&91.2&92.7&92.4&94.4&95.0\\
$L_{1n}(6)$&1.1&4.3&9.2&0.9&5.6&11.3&8.5&25.2&37.7&22.0&41.5&54.8&82.0&88.4&90.7&90.0&92.4&93.2\\
$L_{1n}(9)$&0.9&3.6&9.2&1.1&4.5&9.5&9.5&18.8&30.8&15.8&35.3&47.9&78.2&85.2&88.2&88.4&91.5&92.3\\
$L_{2n}(3)$&1.2&4.1&10.1&1.3&6.2&10.3&16.0&35.2&46.6&32.4&54.5&65.5&87.6&91.2&92.7&92.4&94.4&95.0\\
$L_{2n}(6)$&1.5&5.2&10.5&1.0&5.8&12.1&9.0&26.0&38.7&22.6&42.0&55.5&82.4&88.5&90.8&90.0&92.4&93.2\\
$L_{2n}(9)$&0.9&4.4&11.5&1.3&4.8&10.5&6.1&20.5&32.3&16.9&36.7&49.2&78.6&85.8&88.6&88.4&91.6&92.4\\\\

$T_{1n}(3)$&2.1&6.7&11.9&2.2&6.4&11.6&39.5&60.4&70.1&61.7&77.4&84.5&79.5&85.6&87.4&87.0&90.4&92.1\\
$T_{1n}(6)$&1.7&6.5&11.6&1.6&6.2&11.4&26.3&41.5&54.3&45.9&63.1&72.7&68.3&76.5&79.3&77.9&83.5&86.5\\
$T_{1n}(9)$&1.3&5.8&10.8&1.2&4.8&9.9&14.8&31.2&41.6&32.3&53.7&64.4&60.7&70.7&74.9&72.2&78.4&81.4\\
$T_{2n}(3)$&2.2&7.4&12.8&2.3&6.7&12.7&41.0&60.8&70.9&61.5&78.0&84.5&79.9&85.7&87.8&87.2&91.0&92.1\\
$T_{2n}(6)$&2.2&7.8&13.4&2.0&7.5&12.5&25.1&45.9&57.7&47.5&64.5&74.3&69.3&77.4&80.3&78.6&83.9&87.2\\
$T_{2n}(9)$&2.6&7.5&13.5&1.5&7.0&12.5&18.4&36.7&48.3&35.3&58.0&68.0&63.6&73.2&76.4&73.8&79.4&82.1\\\\
\multicolumn{1}{c}{}&\multicolumn{6}{c}{$\mbox{EGP 4}$}&\multicolumn{6}{c}{$\mbox{EGP 5}$} &\multicolumn{6}{c}{$\mbox{EGP 6}$}\\
\cmidrule(lr){2-7}\cmidrule(lr){8-13}\cmidrule(lr){14-19}
\multicolumn{1}{c}{}&\multicolumn{3}{c}{$n=200$}&\multicolumn{3}{c}{$n=300$}&\multicolumn{3}{c}{$n=200$}&\multicolumn{3}{c}{$n=300$}&\multicolumn{3}{c}{$n=200$}&\multicolumn{3}{c}{$n=300$}\\
\cmidrule(lr){2-7}\cmidrule(lr){8-13}\cmidrule(lr){14-19}
\multicolumn{1}{c}{Tests}&\multicolumn{1}{c}{1\%}&\multicolumn{1}{c}{5\%}&\multicolumn{1}{c}{10\%}&\multicolumn{1}{c}{1\%}&\multicolumn{1}{c}{5\%}&\multicolumn{1}{c}{10\%}&\multicolumn{1}{c}{1\%}&\multicolumn{1}{c}{5\%}&\multicolumn{1}{c}{10\%}&\multicolumn{1}{c}{1\%}&\multicolumn{1}{c}{5\%}&\multicolumn{1}{c}{10\%}&\multicolumn{1}{c}{1\%}&\multicolumn{1}{c}{5\%}&\multicolumn{1}{c}{10\%}&\multicolumn{1}{c}{1\%}&\multicolumn{1}{c}{5\%}&\multicolumn{1}{c}{10\%}\\
\cmidrule(lr){1-1}\cmidrule(lr){2-7}\cmidrule(lr){8-13}\cmidrule(lr){14-19}
$S_{1n}(0)$&0.5&3.7&7.7&0.5&4.4&9.7&76.3&89.4&94.4&92.1&98.5&99.3&92.4&97.8&99.1&98.8&99.8&99.8\\
$S_{1n}(3)$&1.0&4.3&8.9&1.0&4.1&10.1&0.6&3.9&9.0&0.7&4.9&9.1&0.8&4.5&10.3&1.0&4.5&10.1\\
$S_{2n}(3)$&100&100&100&100&100&100&0.7&4.7&9.2&0.6&5.2&9.2&0.7&3.5&7.8&0.6&4.6&9.5\\
$J_{1n}(3)$&0.3&2.5&6.5&0.7&3.9&8.6&33.9&61.2&73.5&61.8&82.0&88.9&56.4&80.2&88.0&86.3&95.3&97.9\\
$J_{1n}(6)$&0.3&1.3&4.1&0.3&3.4&7.0&13.6&40.2&56.6&38.1&64.0&76.6&30.5&57.8&72.2&66.8&85.3&93.0\\
$J_{2n}(3)$&97.1&99.4&99.8&100&100&100&30.1&61.3&74.7&62.0&81.2&89.0&56.6&78.8&87.5&85.9&95.1&98.1\\
$J_{2n}(6)$&83.1&97.0&98.4&99.8&100&100&12.8&38.2&55.3&36.7&63.5&77.1&27.8&57.8&71.7&64.7&84.6&91.9\\\\

$L_{1n}(3)$&86.6&91.2&92.1&93.2&94.4&95.1&51.9&61.1&70.2&66.7&76.4&80.9&49.6&64.8&73.4&68.1&79.5&85.3\\
$L_{1n}(6)$&80.7&87.2&89.4&90.7&93.2&94.3&42.7&57.3&64.3&57.3&69.5&75.6&41.0&57.1&64.1&58.4&72.9&79.0\\
$L_{1n}(9)$&75.1&84.1&86.1&87.9&91.8&92.8&37.6&52.2&59.1&51.6&63.8&70.0&31.8&51.8&59.1&52.7&67.8&74.9\\
$L_{2n}(3)$&87.0&91.4&92.3&93.2&94.4&95.1&52.0&61.2&71.3&66.7&76.5&81.5&49.7&65.0&73.5&68.1&79.6&85.5\\
$L_{2n}(6)$&81.3&87.4&89.7&90.7&93.2&94.3&43.3&58.3&65.0&57.6&69.7&75.8&41.6&57.1&64.5&58.5&73.0&79.1\\
$L_{2n}(9)$&76.6&84.8&87.4&88.0&91.9&93.0&38.1&52.9&60.3&52.0&64.1&70.7&33.1&53.1&60.5&53.4&68.5&75.5\\\\

$T_{1n}(3)$&80.5&85.6&88.1&88.1&90.5&92.2&51.7&59.8&64.4&58.1&67.5&72.2&43.8&55.1&61.2&56.2&65.5&70.1\\
$T_{1n}(6)$&67.2&75.6&79.3&79.8&85.4&87.8&43.2&52.3&57.1&48.2&60.1&65.3&34.7&45.8&52.7&44.5&55.7&61.8\\
$T_{1n}(9)$&60.4&69.0&72.6&71.7&78.5&82.1&37.7&46.7&52.1&41.7&51.8&57.3&29.3&40.4&46.2&40.1&50.4&55.8\\
$T_{2n}(3)$&86.6&91.2&92.1&88.1&90.7&92.3&52.0&59.2&65.2&58.9&67.7&72.6&44.4&55.1&62.8&56.7&65.7&70.4\\
$T_{2n}(6)$&68.7&77.2&81.2&81.0&86.3&88.2&44.9&53.3&57.8&49.5&60.9&66.4&36.9&47.4&54.3&45.7&57.3&62.5\\
$T_{2n}(9)$&63.6&70.8&76.0&73.3&79.9&82.9&40.1&49.0&55.3&43.5&53.8&58.9&32.2&43.7&49.6&42.0&52.5&59.0\\
\hline
\end{tabular}
\end{center}
\end{table}

Overall, our single HSIC-based tests as usual have  good power in detecting dependence at  specific lags, and
our joint HSIC-based tests could be more powerful than other tests in detecting either linear or non-linear dependence.

\section{A real example}
In this section, we study two bivariate time series. The first bivariate time series
consist of two index series from the Russian market and the Indian market: the Russia Trading System Index (RTSI) and the Bombay Stock Exchange Sensitive Index (BSESI).  The second bivariate time series include two Chinese indexes: the ShangHai Securities Composite index (SHSCI) and
 the ShenZhen Index (SZI).  The data are observed on a daily basis (from Monday to Friday), beginning on 8 October 2014, and ending on 29 September 2017.  In all there were 1088 days, missing data due to holidays are removed before the analysis, and hence the final data set include $n=672$ daily observations. The resulting four time series are denoted by   \{RTSI$_t$;  $t=1,\ldots,n\}$,  \{BSESI$_t$;  $t=1,\ldots,n\}$, \{SHSCI$_t$; $t=1,\ldots,n\}$ and \{SZI$_t$;  $t=1,\ldots,n\}$,  respectively.

As usual, we consider the log-return of each data set:
 \begin{flalign*}
&Y_{1t}= \left(
\begin{array}{c}
Y_{1t,1}\\
Y_{1t,2}
\end{array}
\right)=\left(
\begin{array}{ccc}
\log(\mbox{RTSI}_t)-\log(\mbox{RTSI}_{t-1})\\
\log(\mbox{BSESI}_t)-\log(\mbox{BSESI}_{t-1})\\
\end{array}
\right),\\
&Y_{2t}= \left(
\begin{array}{c}
Y_{2t,1}\\
Y_{2t,2}
\end{array}
\right)=\left(
\begin{array}{ccc}
\log(\mbox{SHSCI}_t)-\log(\mbox{SHSCI}_{t-1})\\
\log(\mbox{SZI}_t)-\log(\mbox{SZI}_{t-1})\\
\end{array}
\right).
\end{flalign*}
An investigation on the ACF and PACF of $Y_{1t,1}, Y_{1t,2},Y_{2t,1},Y_{2t,2}$ and their squares
indicates that they do not have a conditional mean structure but a conditional variance structure.
Motivated by this, we use the following BEKK model with Gaussian-QMLE method to fit $Y_{1t}$ and $Y_{2t}$:
\begin{flalign*}
 Y_{st}&=\Sigma_{st}^{1/2}\eta_{st},\\
 \Sigma_{st}&=A_s+B_{s1}^{T}Y_{1t-1}Y_{1t-1}^{T}B_{s1}+\cdots+B_{sp}^{T}Y_{1t-p}Y_{1t-p}^{T}B_{sp}\\
 &+C_{s1}^{T}\Sigma_{st-1}C_{s1}+\cdots+C_{sq}^{T}\Sigma_{st-q}C_{sq}
 \end{flalign*}
for $s=1, 2$, where $A_s=C_{s0}^{T}C_{s0}$ with $C_{s0}$ being a triangular $2\times 2$ matrix, and $B_{s1},\cdots,B_{sp}, C_{s1},\ldots,C_{sq}$ are all $2\times 2$ diagonal matrixes. Table \ref{tab:3} reports the estimates for both fitted models. The p-values of portmanteau tests
$Q(3)$, $Q(6)$ and $Q(9)$ in Ling and Li (1997) are $0.7698,0.5179, 0.5967$ for $Y_{1t}$ and $0.5048, 0.7328, 0.8746$ for $Y_{2t}$. This implies that both fitted BEKK models are adequate.
\begin{table}[!htbp]
\caption{Estimation results for both fitted BEKK models}
\label{tab:3}       
\begin{tabular}{ c c c c c c }
\hline
\multicolumn{1}{c}{$\mbox{Parameters}$}&\multicolumn{2}{c}{Estimates}&\multicolumn{1}{c}{$\mbox{Parameters}$}&\multicolumn{2}{c}{Estimates}\\
\cmidrule(lr){1-1}\cmidrule(lr){2-3}\cmidrule(lr){4-4}\cmidrule(lr){5-6}
$A_1$&$\hat{a}_{1,11}$&0.2832$\times 10^{-3}$          & $A_2$& $\hat{a}_{2,11}$ &0.2528$\times 10^{-5}$ \\
&$\hat{a}_{1,12}$&0.0050$\times 10^{-3}$                     &&$\hat{a}_{2,12}$&0.3856$\times 10^{-5}$\\
&$\hat{a}_{1,22}$&0.0022$\times 10^{-3}$                     &&$\hat{a}_{2,22}$&0.6714$\times 10^{-5}$\\\\

$B_{11}$&$\hat{b}_{11,11}$&0.4662    &$B_{21}$& $\hat{b}_{21,11}$&0.3098 \\
&$\hat{b}_{11,22}$&-0.0619                   &&$\hat{b}_{21,22}$&0.3195\\\\

$B_{12}$&$\hat{b}_{12,11}$&-0.1149                      &$B_{22}$&$\hat{b}_{22,11}$&-0.1264 \\
&$\hat{b}_{12,22}$&0.3357                                   &&$\hat{b}_{22,22}$&-0.0692\\\\

$C_{11}$&$\hat{c}_{11,11}$&0.3569                     &$C_{21}$ &$\hat{c}_{21,11}$&0.6808 \\
&$\hat{c}_{11,22}$&0.2222                                   &&$\hat{c}_{21,22}$&0.6783 \\\\

$C_{12}$&$\hat{c}_{12,11}$&0.5370  &$C_{22}$ &$\hat{c}_{22,11}$&0.6431 \\
&$\hat{c}_{12,22}$&0.9027  &&$\hat{c}_{22,22}$&0.6455 \\
\hline
\multicolumn{6}{l}{$\dagger$\mbox{ Note that} $A_{s}$ \mbox{is a symmetric matrix}, \mbox{and all }$B_{sj}$ \mbox{and} $C_{sj}$ \mbox{are \mbox{diagonal matrixes.}}}
\end{tabular}
\end{table}

Next, we apply our joint HSIC-based tests $J_{sn}(M)$ to
check whether $Y_{1t}$ and $Y_{2t}$ behave independently of each other. As a comparison, we
also consider the tests $L_{sn}(M)$ and $T_{sn}(M)$ for the testing purpose.
 Table \ref{tab:4} reports the $p$-value for  all six tests.
 From Table \ref{tab:4}, we find that except for $J_{2n}(M)$ with $M\geq7$, all examined
 joint HSIC-based tests  $J_{sn}(M)$ convey
 strong evidence that $Y_{1t}$ and $Y_{2t}$ are not independent.
 However, neither $L_{sn}(M)$ nor $T_{sn}(M)$ is able to do this for $M\geq 2$.

To get more information, we further plot the values of the single version of $J_{sn}$, $L_{1n}$ and $T_{1n}$ in Fig 1. That is, Fig 1
plots the values of $S_{sn}(m)$, $L_{1n, s}(m)$, and $T_{1n, s}(m)$ for $m\geq 0$, where
\begin{flalign*}
L_{1n,1}(m)&=n\rho_{\widehat{q}_{1t},\widehat{q}_{2t}}^{2}(m),\,\,\, L_{1n,2}(m)=n\rho_{\widehat{q}_{1t},\widehat{q}_{2t}}^{2}(-m),\\
T_{1n,1}(m)&=n\cdot tr(C_{12}^{T}(m)C_{11}^{-1}(0)C_{12}(m)C_{22}^{-1}(0)), \\
T_{1n,2}(m)&=n\cdot tr(C_{12}^{T}(-m)C_{11}^{-1}(0)C_{12}(-m)C_{22}^{-1}(0)),
\end{flalign*}
and all notations are inherited from Section 5.2. The limiting null distribution of $L_{1n,s}(m)$  is $\chi^2_1$, and that of $T_{1n,s}(m)$ is $\chi_9^{2}$. Similar to
$S_{sn}(m)$,
$L_{1n,s}(m)$ and $T_{1n,s}(m)$ capture the linear dependence between $\eta_{1t}$ and $\eta_{1t+m}$ at the specific lag $m$.
The corresponding single version results for $L_{2n}$ and $T_{2n}$ are similar to those for $L_{1n}$ and $T_{1n}$,
and hence they are not displayed here.

 \begin{table}[!htbp]
\caption{The p-value for all six joint  tests up to lag $M=0,1,\ldots,10$.}
\label{tab:4}       
\begin{tabular}{ c c c c c c c }
\hline
&\multicolumn{6}{c}{Tests}\\
\cmidrule(lr){2-7}
\multicolumn{1}{c}{M}    &  $J_{1n}$& $J_{2n}$ &$L_{1n}$ & $L_{2n}$  & $T_{1n}$& $T_{2n}$ \\
    \cmidrule(lr){1-1}\cmidrule(lr){2-2}\cmidrule(lr){3-3}\cmidrule(lr){4-4}\cmidrule(lr){5-5}\cmidrule(lr){6-6}\cmidrule(lr){7-7}
0  &0.0000    &0.0000    &0.0134   &0.0134   &0.0000   &0.0000\\
1  &0.0000    &0.0000    &0.0428   &0.0428   &0.0125   &0.0124\\
2  &0.0000    &0.0000    &\textbf{0.0881}   &\textbf{0.0879} &\textbf{0.1965}   &\textbf{0.1956}   \\
3  &0.0000    &0.0260    &\textbf{0.0610}   &\textbf{0.0605} &\textbf{0.1055}   &\textbf{0.1035}  \\
4  &0.0000    &0.0040    &\textbf{0.1137}   &\textbf{0.1128} &\textbf{0.2979}   &\textbf{0.2927}   \\
5  &0.0090    &0.0240    &\textbf{0.2111}   &\textbf{0.2095} &\textbf{0.4640}   &\textbf{0.4557}   \\
6  &0.0230    &0.0280    &\textbf{0.2762}   &\textbf{0.2739} &\textbf{0.5958}   &\textbf{0.5851}   \\
7  &0.0220    &\textbf{0.0720}    &\textbf{0.3315}   &\textbf{0.3282} &\textbf{0.7093}   &\textbf{0.6972}   \\
8  &0.0280    &\textbf{0.0730}    &\textbf{0.4079}  &\textbf{0.4037} &\textbf{0.6708}   &\textbf{0.6540}  \\
9  &0.0450    &\textbf{0.0830}    &\textbf{0.4491}   &\textbf{0.4437}  &\textbf{0.7645}   &\textbf{0.7475}   \\
10 &0.0230    &\textbf{0.1040}   &\textbf{0.5761}   &\textbf{0.5706} &\textbf{0.8359}   &\textbf{0.8199} \\
\hline
\multicolumn{7}{l}{$\dagger$ A p-value larger than 5\% is in boldface.}
\end{tabular}
\end{table}

From Fig 1, we first find that all single tests indicate a strong contemporaneously causal relationship between the Chinese market and the Russian and Indian (R\&I) market. Second, $S_{1n}(1)$ implies that the R\&I market has significant influence on the Chinese market one day later, while according to $S_{2n}(3)$ (or $S_{2n}(10)$), the impact of the Chinese market to the R\&I market appears after three (or ten) days.
These findings demonstrate an asymmetric causal relationship between two markets.
Since none of examined $L_{1n,s}(m)$ and $T_{1n,s}(m)$ can detect a causal relationship for $m\geq1$,
the contemporaneous causal relationship mainly results in the significance of $L_{sn}(1)$ and $T_{sn}(1)$ in Table 4, and
the lagged causal relationship is possible to be non-linear.
As the R\&I market has a higher degree of
globalization and marketization,  it could have a quicker impact to other economies. On the contrary,
the Chinese market is more localized, and its influence to other economies tends to be slower but can last for a
longer term. This long-term effect may be caused by ``the Belt and Road Initiatives'' program raised by Chinese government
since 2015. Hence, the asymmetric phenomenon between two markets seems reasonable, and
it may help the government to make more efficient policy and the investors to design more useful investment strategies.

\begin{figure}[!htbp]
\begin{center}
\includegraphics[width=28pc,height=16pc]{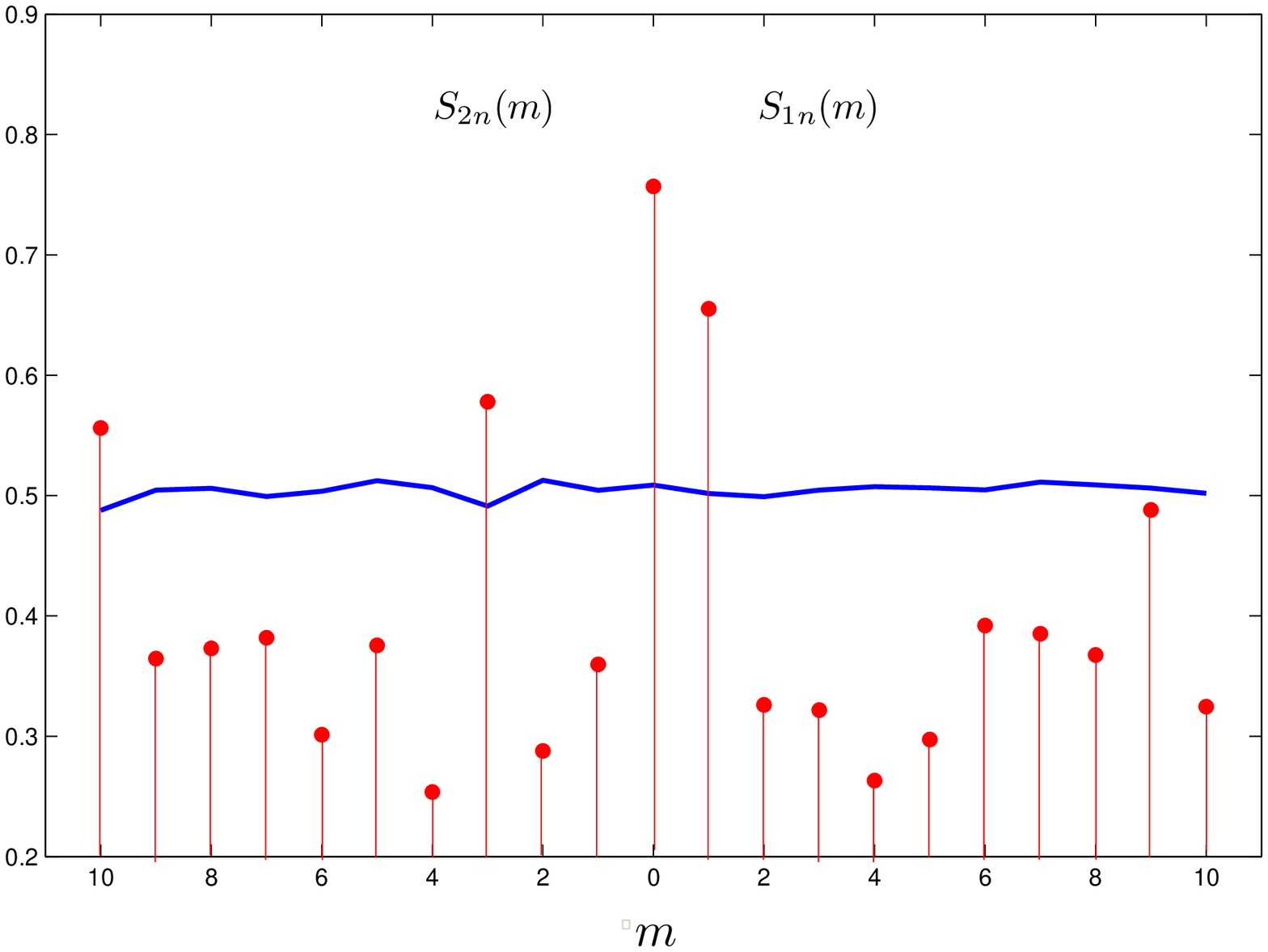}
\includegraphics[width=28pc,height=16pc]{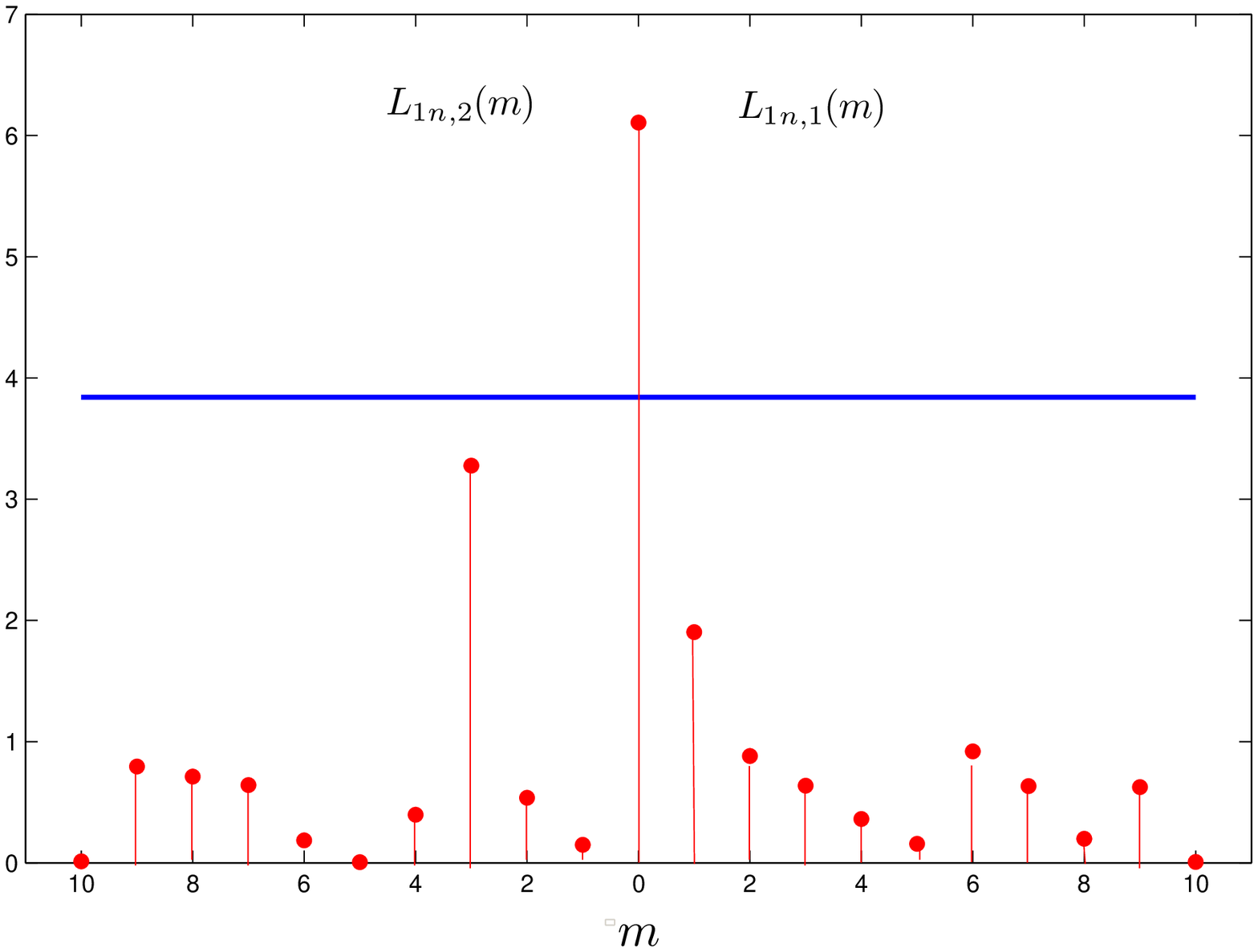}
\includegraphics[width=28pc,height=16pc]{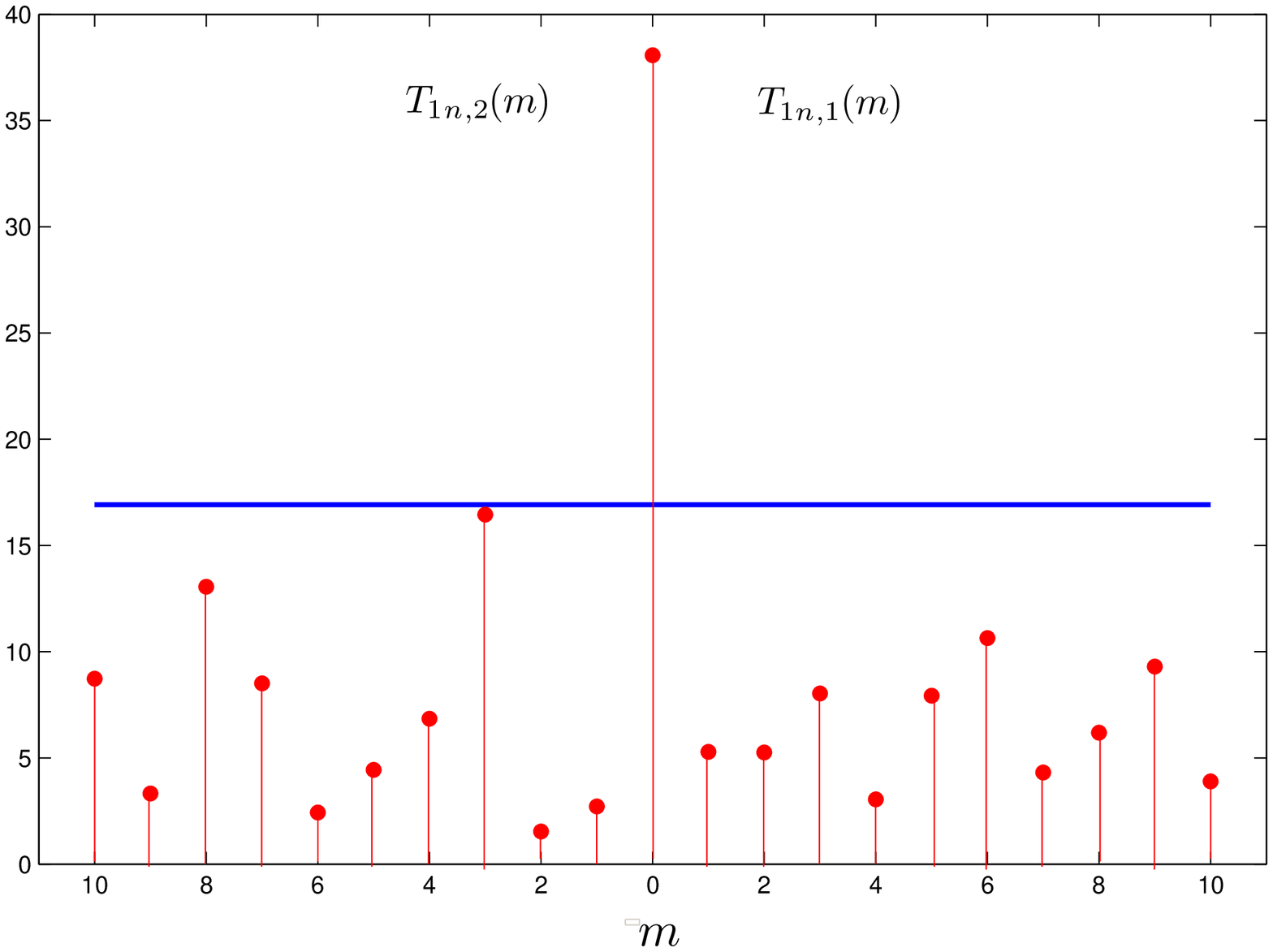}
\caption{\label{figure1}The values of single tests $S_{1n}(m)$, $L_{1n,1}(m)$ and $T_{1n,1}(m)$ (right panel)  across $m$, and the values of single tests $S_{2n}(m)$, $L_{1n,2}(m)$ and $T_{1n,2}(m)$ (left panel) across $m$.
The solid lines are  95\% one-sided confidence bounds of the tests.}
\end{center}
\end{figure}

%

\section{Concluding remarks}

In this paper, we apply the HSIC principle to derive some novel one-sided omnibus tests for
detecting  independence between two multivariate stationary time
series. The resulting HSIC-based tests have asymptotical Gaussian representation under the null hypothesis, and they are shown to be consistent.
A residual bootstrap method is used to obtain the critical values for our HSIC-based tests, and its validity is justified.
Unlike the existing cross-correlation-based tests for linear dependence, our HSIC-based tests look for the general dependence between two
un-observable innovation vectors, and hence they can give investigators more complete
information on the causal relationship between two time series.
The importance of our HSIC-based tests is illustrated by simulation results and real data analysis.
Due to the generality of the HSIC method, the methodology developed in this paper may be applied to many other
important testing problems such as testing for model adequacy (Davis et al. 2016),
testing for independence among multi-dynamic systems (Pfister et al. 2017), or testing for independence in high dimensional systems (Yao et al. 2017). We leave these interesting topics as potential future study.

\appendix
\section*{Appendix: Proofs}

This appendix provides the proofs of all lemmas and theorems. To facilitate it,
the results of V-statistics are needed below, and they can be found in Hoeffding (1948) and Lee (1990) for the
i.i.d. case and Yoshihara (1976) and Denker and Keller (1983) for the mixing case.

\vspace{4mm}

\textsc{Proof of Lemma \ref{lem3.1}.}\,\,Denote
$\widehat{z}_{ijqr}=\widehat{k}_{ij}\widehat{l}_{qr}$.  By Taylor's
expansion,
\begin{flalign}
\widehat{z}_{ijqr}&=z_{ijqr}^{(0)}+(\widehat{\eta}_{ijqr}-\eta_{ijqr})^{T}W_{ijqr}\nonumber\\
&\quad+\frac{1}{2}
(\widehat{\eta}_{ijqr}-\eta_{ijqr})^{T}H_{ijqr}^{\dag}(\widehat{\eta}_{ijqr}-\eta_{ijqr})\nonumber\\
&=z_{ijqr}^{(0)}+(\widehat{\eta}_{ijqr}-\eta_{ijqr})^{T}W_{ijqr}\nonumber\\
&\quad+\frac{1}{2}
(\widehat{\eta}_{ijqr}-\eta_{ijqr})^{T}H_{ijqr}(\widehat{\eta}_{ijqr}-\eta_{ijqr})+R_{ijqr}^{(1)},\label{A2}
\end{flalign}
where $z_{ijqr}^{(0)}=k_{ij}l_{qr}$,
$\widehat{\eta}_{ijqr}=(\widehat{\eta}_{1i}^{T},\widehat{\eta}_{1j}^{T},\widehat{\eta}_{2q+m}^{T},\widehat{\eta}_{2r+m}^{T})^{T}$,
$\eta_{ijqr}=(\eta_{1i}^{T},\eta_{1j}^{T},\eta_{2q+m}^{T},$ $\eta_{2r+m}^{T})^{T}$,
$W_{ijqr}=W(\eta_{ijqr})$, $H_{ijqr}=H(\eta_{ijqr})$,
$H_{ijqr}^{\dag}=H(\eta_{ijqr}^{\dag})$, $\eta_{ijqr}^{\dag}$
lies between $\eta_{ijqr}$ and $\widehat{\eta}_{ijqr}$, and
$$R_{ijqr}^{(1)}=(\widehat{\eta}_{ijqr}-\eta_{ijqr})^{T}\left(H_{ijqr}^{\dag}-H_{ijqr}\right)(\widehat{\eta}_{ijqr}-\eta_{ijqr}).$$
Here,
$W:\mathcal{R}^{d_{1}}\times\mathcal{R}^{d_{1}}\times\mathcal{R}^{d_{2}}\times\mathcal{R}^{d_{2}}\to
\mathcal{R}^{(2d_{1}+2d_{2})\times1}$ such that
\begin{flalign*}
&W(u,u',v,v')=\\
&\Big(k_{x}(u,u')^{T}l(v,v'),k_{y}(u,u')^{T}l(v,v'),k(u,u')l_{x}(v,v')^{T},k(u,u')l_{y}(v,v')^{T}\Big)^{T},
\end{flalign*}
and $H:
\mathcal{R}^{d_{1}}\times\mathcal{R}^{d_{1}}\times\mathcal{R}^{d_{2}}\times\mathcal{R}^{d_{2}}\to\mathcal{R}^{2d_{1}+2d_{2}}\times
\mathcal{R}^{2d_{1}+2d_{2}}$ such that
\begin{flalign*}
&H(u,u',v,v')=\\
&\left(\begin{array}{cccc}
k_{xx}(u,u')l(v,v') & k_{xy}(u,u')l(v,v') &
k_{x}(u,u')l_{x}(v,v')^{T} &
k_{x}(u,u')l_{y}(v,v')^{T}\\
* & k_{yy}(u,u')l(v,v') & k_{y}(u,u')l_{x}(v,v')^{T} &
k_{y}(u,u')l_{y}(v,v')^{T}\\
* & * & k(u,u')l_{xx}(v,v') & k(u,u')l_{xy}(v,v')\\
* & * &
* & k(u,u')l_{yy}(v,v')
\end{array}
\right)
\end{flalign*}
is a symmetric matrix.

Next, let $\theta=(\theta_{1}^{T},\theta_{2}^{T})^{T}$ and
$\widehat{\theta}_{n}=(\widehat{\theta}_{1n}^{T},\widehat{\theta}_{2n}^{T})^{T}$,
and denote
$$G_{ijqr}(\theta)=\Big(g_{1i}(\theta_{1})^{T},g_{1j}(\theta_{1})^{T},
g_{2q+m}(\theta_{2})^{T},g_{2r+m}(\theta_{2})^{T}\Big)^{T},$$
where $g_{st}(\theta_{s})$ is defined as in Assumption \ref{asm2.2}. By Taylor's expansion again, we have
\begin{flalign}
\widehat{\eta}_{ijqr}-\eta_{ijqr}&=\overline{R}_{ijqr}^{(2)}
+\frac{\p
G_{ijqr}(\theta^{\dag})}{\p\theta^{T}}(\widehat{\theta}_{n}-\theta_{0}),\label{A3}
\end{flalign}
where $\overline{R}_{ijqr}^{(2)}=
(\widehat{R}_{1i}(\widehat{\theta}_{1n})^{T},\widehat{R}_{1j}(\widehat{\theta}_{1n})^{T},
\widehat{R}_{2q+m}(\widehat{\theta}_{2n})^{T},\widehat{R}_{2r+m}(\widehat{\theta}_{2n})^{T})^{T}$,
$\widehat{R}_{st}(\theta_{s})$ is defined as in Assumption \ref{asm2.4}, and
$\theta^{\dag}$ lies between $\theta_{0}$ and $\widehat{\theta}_{n}$.
For the second term in (\ref{A3}), we rewrite it as
\begin{flalign}
\frac{\p
G_{ijqr}(\theta^{\dag})}{\p\theta^{T}}(\widehat{\theta}_{n}-\theta_{0})=
\overline{R}_{ijqr}^{(3)}+\frac{\p
G_{ijqr}(\theta_{0})}{\p\theta^{T}}(\widehat{\theta}_{n}-\theta_{0}), \label{A4}
\end{flalign}
where $\overline{R}_{ijqr}^{(3)}=\big[\frac{\p
G_{ijqr}(\theta^{\dag})}{\p\theta^{T}}-\frac{\p
G_{ijqr}(\theta_{0})}{\p\theta^{T}}\big](\widehat{\theta}_{n}-\theta_{0})$.

Now, by (\ref{A2})-(\ref{A4}), it follows that
\begin{flalign}
\widehat{z}_{ijqr}=z_{ijqr}^{(0)}+(\widehat{\theta}_{n}-\theta_{0})^{T}z_{ijqr}^{(1)}+\frac{1}{2}
(\widehat{\theta}_{n}-\theta_{0})^{T}z_{ijqr}^{(2)}(\widehat{\theta}_{n}-\theta_{0})+R_{ijqr},\label{A5}
\end{flalign}
where
$z_{ijqr}^{(1)}=\frac{\p
G_{ijqr}(\theta_{0})}{\p\theta}W_{ijqr}$, $z_{ijqr}^{(2)}=\frac{\p
G_{ijqr}(\theta_{0})}{\p\theta}H_{ijqr}\frac{\p
G_{ijqr}(\theta_{0})}{\p\theta^{T}}$, and $R_{ijqr}=R_{ijqr}^{(1)}+R_{ijqr}^{(2)}+R_{ijqr}^{(3)}+R_{ijqr}^{(4)}$
with
\begin{flalign*}
R_{ijqr}^{(2)}&=\left(\overline{R}_{ijqr}^{(2)}+\overline{R}_{ijqr}^{(3)}\right)^{T}W_{ijqr},\\
R_{ijqr}^{(3)}&=\frac{1}{2}\left(\overline{R}_{ijqr}^{(2)}+\overline{R}_{ijqr}^{(3)}\right)^{T}H_{ijqr}
\left(\overline{R}_{ijqr}^{(2)}+\overline{R}_{ijqr}^{(3)}\right),\\
R_{ijqr}^{(4)}&=(\widehat{\theta}_{n}-\theta_{0})^{T}\frac{\p
G_{ijqr}(\theta_{0})}{\p\theta}H_{ijqr}
\left(\overline{R}_{ijqr}^{(2)}+\overline{R}_{ijqr}^{(3)}\right).
\end{flalign*}
By (\ref{A5}), it entails that
\begin{flalign}
S_{1n}(m)&=S_{1n}^{(0)}(m)+(\widehat{\theta}_{n}-\theta_{0})^{T}S_{1n}^{(1)}(m)+
\frac{1}{2}(\widehat{\theta}_{n}-\theta_{0})^{T}S_{1n}^{(2)}(m)(\widehat{\theta}_{n}-\theta_{0})\nonumber\\
&\quad+R_{1n}(m),\label{A6}
\end{flalign}
where
$$S_{1n}^{(p)}(m)=\frac{1}{N^{2}}\sum_{i,j}z_{ijij}^{(p)}+\frac{1}{N^{4}}
\sum_{i,j,q,r}z_{ijqr}^{(p)}
-\frac{2}{N^{3}}\sum_{i,j,q}z_{ijiq}^{(p)}$$ for $p\in\{0,1,2\}$,
and
\begin{flalign}
R_{1n}(m)=\frac{1}{N^{2}}\sum_{i,j}R_{ijij}+\frac{1}{N^{4}}\sum_{i,j,q,r}R_{ijqr}
-\frac{2}{N^{3}}\sum_{i,j,q}R_{ijiq}\label{A7}
\end{flalign}
is the remainder term.

Furthermore, simple algebra shows that
\begin{flalign}
(\widehat{\theta}_{n}-\theta_{0})^{T}z_{ijqr}^{(1)}&=\zeta_{1n}^{T}\overline{k}_{ij}l_{qr}+
\zeta_{2n}^{T}k_{ij}\overline{l}_{qr},\label{A8}\\
(\widehat{\theta}_{n}-\theta_{0})^{T}z_{ijqr}^{(2)}(\widehat{\theta}_{n}-\theta_{0})&=
\zeta_{1n}^{T}\widecheck{k}_{ij}l_{qr}\zeta_{1n}
+
\zeta_{2n}^{T}k_{ij}\widecheck{l}_{qr}\zeta_{2n}+\zeta_{1n}^{T}\left(2\overline{k}_{ij}\overline{l}_{qr}^{T}\right)\zeta_{2n},\label{A9}
\end{flalign}
where  $\overline{k}_{ij}$, $\overline{l}_{ij}$,
$\widecheck{k}_{ij}$, and $\widecheck{l}_{ij}$ are defined  in
(3.1)-(3.4), respectively.
Finally, the conclusion holds by (\ref{A6}) and
(\ref{A8})-(\ref{A9}). This
completes the proof. $\hfill\square$


\vspace{4mm}

\textsc{Proof of Lemma \ref{lem3.2}.}\,\,Without loss of generality, we only prove the results for $m=0$, under which
$N=n$, and $\eta_{t}^{(0)}$ and $\varsigma_{t}^{(0)}$ are denoted by $\eta_{t}:=(\eta_{1t},\eta_{2t})$ and $\varsigma_{t}:=\left(\eta_{1t},\frac{\p
g_{1t}(\theta_{10})}{\p\theta_{1}},\eta_{2t},\frac{\p g_{2t}(\theta_{20})}{\p\theta_{2}}\right)$, respectively, for
notational ease.

(i) Denote $x_{1}=(x_{11},x_{21})$ for
$x_{11}\in\mathcal{R}^{d_{1}}$ and $x_{21}\in\mathcal{R}^{d_{2}}$.
Then, we rewrite
\begin{flalign*}
h^{(0)}_{0}(x_{1},\eta_{2},\eta_{3},\eta_{4})&=\frac{1}{4!}\left[\sum_{t=1,(u,v,w)}^{(2,3,4)}z_{1,uvw}^{(0)}(x_{1})+
\sum_{u=1,(t,v,w)}^{(2,3,4)}z_{2,tvw}^{(0)}(x_{1})\right.\\
&\quad\left.+\sum_{v=1,(t,u,w)}^{(2,3,4)}z_{3,tuw}^{(0)}(x_{1})+\sum_{w=1,(t,u,v)}^{(2,3,4)}
z_{4,tuv}^{(0)}(x_{1})\right]\\
&=:\frac{1}{4!}\left[\Delta_{1}^{(0)}+\Delta_{2}^{(0)}+\Delta_{3}^{(0)}+\Delta_{4}^{(0)}\right],
\end{flalign*}
where
\begin{flalign*}
z_{1,uvw}^{(0)}(x_{1})&=k(x_{11},\eta_{1u})\left[l(x_{21},\eta_{2u})+l(\eta_{2v},\eta_{2w})-2l(x_{21},\eta_{2v})\right],\\
z_{2,tvw}^{(0)}(x_{1})&=k(\eta_{1t},x_{11})\left[l(\eta_{2t},x_{21})+l(\eta_{2v},\eta_{2w})-2l(\eta_{2t},\eta_{2v})\right],\\
z_{3,tuw}^{(0)}(x_{1})&=k(\eta_{1t},\eta_{1u})\left[l(\eta_{2t},\eta_{2u})+l(x_{21},\eta_{2w})-2l(\eta_{2t},x_{21})\right],\\
z_{4,tuv}^{(0)}(x_{1})&=k(\eta_{1t},\eta_{1u})\left[l(\eta_{2t},\eta_{2u})+l(\eta_{2v},x_{21})-2l(\eta_{2t},\eta_{2v})\right].
\end{flalign*}
By the symmetry of $k$ and $l$, the stationarity of $\eta_{1t}$ and
$\eta_{2t}$, and the independence of $\{\eta_{1t}\}$ and $\{\eta_{2t}\}$
under $H_{0}$, simple algebra shows that
\begin{flalign*}
E\Delta_{1}^{(0)}&=6E\left[k(x_{11},\eta_{11})\right]\times
E\left[l(\eta_{21},\eta_{22})-l(x_{21},\eta_{21})\right],\\
E\Delta_{2}^{(0)}&=6E\left[k(x_{11},\eta_{11})\right]\times
E\left[l(x_{21},\eta_{21})-l(\eta_{21},\eta_{22})\right],\\
E\Delta_{3}^{(0)}&=6E\left[k(\eta_{11},\eta_{12})\right]\times
E\left[l(\eta_{21},\eta_{22})-l(x_{21},\eta_{21})\right],\\
E\Delta_{4}^{(0)}&=6E\left[k(\eta_{11},\eta_{12})\right]\times
E\left[l(x_{21},\eta_{21})-l(\eta_{21},\eta_{22})\right].
\end{flalign*}
Hence, it follows that under $H_{0}$, $E[h^{(0)}_{0}(x_{1},\eta_{2},\eta_{3},\eta_{4})]=0$ for all $x_{1}$. This completes the proof of (i).

(ii) We only consider the proof for the case that $a=b=1$, since the proofs of other cases are similar.
Denote $x_{1}=(x_{11},y_{11},x_{21},y_{21})$ for
$x_{11}\in\mathcal{R}^{d_{1}}$, $y_{11}\in\mathcal{R}^{p_{1}\times d_{1}}$, $x_{21}\in\mathcal{R}^{d_{2}}$, and $y_{21}\in
\mathcal{R}^{p_{2}\times d_{2}}$.
Then, we rewrite
\begin{flalign*}
h^{(11)}_{0}(x_{1},\varsigma_{2},\varsigma_{3},\varsigma_{4})&=\frac{1}{4!}\left[\sum_{t=1,(u,v,w)}^{(2,3,4)}z_{1,uvw}^{(11)}(x_{1})+
\sum_{u=1,(t,v,w)}^{(2,3,4)}z_{2,tvw}^{(11)}(x_{1})\right.\\
&\quad\left.+\sum_{v=1,(t,u,w)}^{(2,3,4)}z_{3,tuw}^{(11)}(x_{1})+\sum_{w=1,(t,u,v)}^{(2,3,4)}
z_{4,tuv}^{(11)}(x_{1})\right]\\
&=:\frac{1}{4!}\left[\Delta_{1}^{(11)}+\Delta_{2}^{(11)}+\Delta_{3}^{(11)}+\Delta_{4}^{(11)}\right],
\end{flalign*}
where
\begin{flalign*}
z_{1,uvw}^{(11)}(x_{1})&=\left[y_{11}k_{x}(x_{11},\eta_{1u})+\frac{\p g_{1u}(\theta_{10})}{\p\theta_{1}}k_{x}(\eta_{1u},x_{11})\right]\\
&\quad\quad\times\left[l(x_{21},\eta_{2u})+l(\eta_{2v},\eta_{2w})-2l(x_{21},\eta_{2v})\right],\\
z_{2,tvw}^{(11)}(x_{1})&=\left[\frac{\p g_{1t}(\theta_{10})}{\p\theta_{1}}k_{x}(\eta_{1t},x_{11})+y_{11}k_{x}(x_{11},\eta_{1t})\right]\\
&\quad\quad\times\left[l(\eta_{2t},x_{21})+l(\eta_{2v},\eta_{2w})-2l(\eta_{2t},\eta_{2v})\right],\\
z_{3,tuw}^{(11)}(x_{1})&=\left[\frac{\p g_{1t}(\theta_{10})}{\p\theta_{1}}k_{x}(\eta_{1t},\eta_{1u})+\frac{\p g_{1u}(\theta_{10})}{\p\theta_{1}}k_{x}(\eta_{1u},\eta_{1t})\right]\\
&\quad\quad\times\left[l(\eta_{2t},\eta_{2u})+l(x_{21},\eta_{2w})-2l(\eta_{2t},x_{21})\right],\\
z_{4,tuv}^{(11)}(x_{1})&=\left[\frac{\p g_{1t}(\theta_{10})}{\p\theta_{1}}k_{x}(\eta_{1t},\eta_{1u})+\frac{\p g_{1u}(\theta_{10})}{\p\theta_{1}}k_{x}(\eta_{1u},\eta_{1t})\right]\\
&\quad\quad\times\left[l(\eta_{2t},\eta_{2u})+l(\eta_{2v},x_{21})-2l(\eta_{2t},\eta_{2v})\right].
\end{flalign*}
Here, we have used the fact that
$k_{y}(c,d)=k_{x}(d,c)$ by the
symmetry of $k$. By the stationarity of $\eta_{1t}$ and $\eta_{2t}$, and the
independence of $\{\eta_{1t}\}$ and $\{\eta_{2t}\}$ under $H_{0}$, simple
algebra shows that
\begin{flalign*}
E\Delta_{1}^{(11)}&=-E\Delta_{2}^{(11)}\\
&=\left\{y_{11}Ek_{x}(x_{11},\eta_{11})+E\left[\frac{\p g_{11}(\theta_{10})}{\p\theta_{1}}k_{x}(\eta_{11},x_{11})\right]\right\}\\
&\quad\quad\times\left[4El(\eta_{21},\eta_{22})+2El(\eta_{21},\eta_{23})-6El(x_{21},\eta_{21})\right],\\
E\Delta_{3}^{(11)}&=-E\Delta_{4}^{(11)}\\
&=4E\left[\frac{\p g_{11}(\theta_{10})}{\p\theta_{1}}k_{x}(\eta_{11},\eta_{12})
+\frac{\p g_{12}(\theta_{10})}{\p\theta_{1}}k_{x}(\eta_{12},\eta_{11})\right]\\
&\quad\quad\times\left[El(\eta_{21},\eta_{22})-El(x_{21},\eta_{21})\right]\\
&\quad+2E\left[\frac{\p g_{11}(\theta_{10})}{\p\theta_{1}}k_{x}(\eta_{11},\eta_{13})+\frac{\p g_{13}(\theta_{10})}{\p\theta_{1}}k_{x}(\eta_{13},\eta_{11})\right]\\
&\quad\quad\times\left[El(\eta_{21},\eta_{23})-El(x_{21},\eta_{21})\right].
\end{flalign*}
Hence, it follows that under $H_{0}$,
$E[h^{(11)}_{0}(x_{1},\varsigma_{2},\varsigma_{3},\varsigma_{4})]=0$ for all $x_{1}$. This completes the proof of (ii).

(iii) Denote $x_{1}=(x_{11},y_{11},x_{21},y_{21})$ for
$x_{11}\in\mathcal{R}^{d_{1}}$, $y_{11}\in\mathcal{R}^{p_{1}\times d_{1}}$, $x_{21}\in\mathcal{R}^{d_{2}}$, and $y_{21}\in
\mathcal{R}^{p_{2}\times d_{2}}$.
Then, we rewrite
\begin{flalign*}
h^{(23)}_{0}(x_{1},\varsigma_{2},\varsigma_{3},\varsigma_{4})&=\frac{1}{4!}\left[\sum_{t=1,(u,v,w)}^{(2,3,4)}z_{1,uvw}^{(23)}(x_{1})+
\sum_{u=1,(t,v,w)}^{(2,3,4)}z_{2,tvw}^{(23)}(x_{1})\right.\\
&\quad\left.+\sum_{v=1,(t,u,w)}^{(2,3,4)}z_{3,tuw}^{(23)}(x_{1})+\sum_{w=1,(t,u,v)}^{(2,3,4)}
z_{4,tuv}^{(23)}(x_{1})\right]\\
&=:\frac{1}{4!}\left[\Delta_{1}^{(23)}+\Delta_{2}^{(23)}+\Delta_{3}^{(23)}+\Delta_{4}^{(23)}\right],
\end{flalign*}
where
\begin{flalign*}
z_{1,uvw}^{(23)}(x_{1})&=\left[y_{11}k_{x}(x_{11},\eta_{1u})
+\frac{\p g_{1u}(\theta_{10})}{\p\theta_{1}}k_{x}(\eta_{1u},x_{11})\right]\\
&\quad\quad\times\left[y_{21}l_{x}(x_{21},\eta_{2u})
+\frac{\p g_{2u}(\theta_{20})}{\p\theta_{2}}l_{x}(\eta_{2u},x_{21})
+\frac{\p g_{2v}(\theta_{20})}{\p\theta_{2}}l_{x}(\eta_{2v},\eta_{2w})\right.\\
&\quad\quad\quad\left.
+\frac{\p g_{2w}(\theta_{20})}{\p\theta_{2}}l_{x}(\eta_{2w},\eta_{2v})
-2y_{21}l_{x}(x_{21},\eta_{2v})
-2\frac{\p g_{2v}(\theta_{20})}{\p\theta_{2}}l_{x}(\eta_{2v},x_{21})\right],\\
z_{2,tvw}^{(23)}(x_{1})&=\left[y_{11}k_{x}(x_{11},\eta_{1t})
+\frac{\p g_{1t}(\theta_{10})}{\p\theta_{1}}k_{x}(\eta_{1t},x_{11})\right]\\
&\quad\quad\times\left[y_{21}l_{x}(x_{21},\eta_{2t})
+\frac{\p g_{2t}(\theta_{20})}{\p\theta_{2}}l_{x}(\eta_{2t},x_{21})
+\frac{\p g_{2v}(\theta_{20})}{\p\theta_{2}}l_{x}(\eta_{2v},\eta_{2w})\right.\\
&\quad\quad\quad\left.
+\frac{\p g_{2w}(\theta_{20})}{\p\theta_{2}}l_{x}(\eta_{2w},\eta_{2v})
-2\frac{\p g_{2t}(\theta_{20})}{\p\theta_{2}}l_{x}(\eta_{2t},\eta_{2v})
-2\frac{\p g_{2v}(\theta_{20})}{\p\theta_{2}}l_{x}(\eta_{2v},\eta_{2t})\right],
\end{flalign*}
\begin{flalign*}
z_{3,tuw}^{(23)}(x_{1})&=\left[\frac{\p g_{1t}(\theta_{10})}{\p\theta_{1}}k_{x}(\eta_{1t},\eta_{1u})+\frac{\p g_{1u}(\theta_{10})}{\p\theta_{1}}k_{x}(\eta_{1u},\eta_{1t})\right]\\
&\quad\quad\times\left[y_{21}l_{x}(x_{21},\eta_{2w})
+\frac{\p g_{2t}(\theta_{20})}{\p\theta_{2}}l_{x}(\eta_{2t},\eta_{2u})
+\frac{\p g_{2u}(\theta_{20})}{\p\theta_{2}}l_{x}(\eta_{2u},\eta_{2t})\right.\\
&\quad\quad\quad\left.
+\frac{\p g_{2w}(\theta_{20})}{\p\theta_{2}}l_{x}(\eta_{2w}, x_{21})
-2y_{21}l_{x}(x_{21}, \eta_{2t})
-2\frac{\p g_{2t}(\theta_{20})}{\p\theta_{2}}l_{x}(\eta_{2t}, x_{21})\right],\\
z_{4,tuv}^{(23)}(x_{1})&=\left[\frac{\p g_{1t}(\theta_{10})}{\p\theta_{1}}k_{x}(\eta_{1t},\eta_{1u})+\frac{\p g_{1u}(\theta_{10})}{\p\theta_{1}}k_{x}(\eta_{1u},\eta_{1t})\right]\\
&\quad\quad\times\left[y_{21}l_{x}(x_{21},\eta_{2v})
+\frac{\p g_{2t}(\theta_{20})}{\p\theta_{2}}l_{x}(\eta_{2t},\eta_{2u})
+\frac{\p g_{2u}(\theta_{20})}{\p\theta_{2}}l_{x}(\eta_{2u},\eta_{2t})\right.\\
&\quad\quad\quad\left.
+\frac{\p g_{2v}(\theta_{20})}{\p\theta_{2}}l_{x}(\eta_{2v}, x_{21})
-2\frac{\p g_{2t}(\theta_{20})}{\p\theta_{2}}l_{x}(\eta_{2t},\eta_{2v})
-2\frac{\p g_{2v}(\theta_{20})}{\p\theta_{2}}l_{x}(\eta_{2v},\eta_{2t})\right].
\end{flalign*}
By the stationarity of $\eta_{1t}$ and $\eta_{2t}$, and the
independence of $\{\eta_{1t}\}$ and $\{\eta_{2t}\}$ under $H_{0}$, simple
algebra shows that
\begin{flalign*}
E\Delta_{1}^{(23)}&=-E\Delta_{2}^{(23)}\\
&=\left\{y_{11}Ek_{x}(x_{11},\eta_{11})+E\left[\frac{\p g_{11}(\theta_{10})}{\p\theta_{1}}k_{x}(\eta_{11},x_{11})\right]\right\}\\
&\quad\quad\times\left\{-6y_{21}El_{x}(x_{21},\eta_{21})
-6E\left[\frac{\p g_{21}(\theta_{20})}{\p\theta_{2}}l_{x}(\eta_{21},x_{21})\right]\right.\\
&\quad\quad\quad\quad\left.+4E\left[\frac{\p g_{21}(\theta_{20})}{\p\theta_{2}}l_{x}(\eta_{21},\eta_{22})\right]
+2E\left[\frac{\p g_{21}(\theta_{20})}{\p\theta_{2}}l_{x}(\eta_{21},\eta_{23})\right]\right.\\
&\quad\quad\quad\quad\left.+4E\left[\frac{\p g_{22}(\theta_{20})}{\p\theta_{2}}l_{x}(\eta_{22},\eta_{21})\right]
+2E\left[\frac{\p g_{23}(\theta_{20})}{\p\theta_{2}}l_{x}(\eta_{23},\eta_{21})\right]\right\},\\
E\Delta_{3}^{(11)}&=-E\Delta_{4}^{(11)}+\Upsilon\\
&=4E\left[\frac{\p g_{11}(\theta_{10})}{\p\theta_{1}}k_{x}(\eta_{11},\eta_{12})+\frac{\p g_{12}(\theta_{10})}{\p\theta_{1}}k_{x}(\eta_{12},\eta_{11})\right]\\
&\quad\quad\times\left\{E\left[\frac{\p g_{21}(\theta_{20})}{\p\theta_{2}}l_{x}(\eta_{21},\eta_{22})\right]
-E\left[\frac{\p g_{21}(\theta_{20})}{\p\theta_{2}}l_{x}(\eta_{21},x_{21})\right]\right.\\
&\quad\quad\quad\quad\left.+E\left[\frac{\p g_{22}(\theta_{20})}{\p\theta_{2}}l_{x}(\eta_{22},\eta_{21})\right]
-y_{21}El_{x}(x_{21},\eta_{21})\right\}\\
&\quad+2E\left[\frac{\p g_{11}(\theta_{10})}{\p\theta_{1}}k_{x}(\eta_{11},\eta_{13})+\frac{\p g_{13}(\theta_{10})}{\p\theta_{1}}k_{x}(\eta_{13},\eta_{11})\right]\\
&\quad\quad\times\left\{E\left[\frac{\p g_{21}(\theta_{20})}{\p\theta_{2}}l_{x}(\eta_{21},\eta_{23})\right]
-E\left[\frac{\p g_{21}(\theta_{20})}{\p\theta_{2}}l_{x}(\eta_{21},x_{21})\right]\right.\\
&\quad\quad\quad\quad\left.+E\left[\frac{\p g_{23}(\theta_{20})}{\p\theta_{2}}l_{x}(\eta_{23},\eta_{21})\right]
-y_{21}El_{x}(x_{21},\eta_{21})\right\}.
\end{flalign*}
Hence, it follows that under $H_{0}$,
$E[h^{(23)}_{0}(x_{1},\varsigma_{2},\varsigma_{3},\varsigma_{4})]=\Upsilon$ for all $x_{1}$.
This completes the proof of (iii). $\hfill\square$



\vspace{4mm}

\textsc{Proof of Lemma \ref{lem3.3}.}\,\,
Let
$\mathcal{F}_{i}=\sigma(\mathcal{F}_{1i},\mathcal{F}_{2i})$.
Under $H_{0}$, it is not hard to see that $E(\mathcal{T}_{1i}|\mathcal{F}_{i-1})=E(\mathcal{T}_{1i})=0$
by Lemma \ref{lem3.2}(i). Since $E(\mathcal{T}_{2i}|\mathcal{F}_{i-1})=0$ by Assumption \ref{asm2.3}, it follows that
$E(\mathcal{T}_{i}|\mathcal{F}_{i-1})=0$. Moreover,
by Assumptions \ref{asm2.3} and \ref{asm2.5},
it is straightforward to see that  $E\|\mathcal{T}_{i}\|^{2}<\infty$.
By the central limit theorem for martingale difference sequence (see Corollary 5.26 in White (2001)), it follows that
$\mathcal{T}_{n}\to_{d} \mathcal{T}$
as $n\to\infty$, where $\mathcal{T}$ is a multivariate normal distribution with covariance
matrix $\overline{\mathcal{T}}=\lim_{n\to\infty}var(\mathcal{T}_{n})=E(\mathcal{T}_{1}\mathcal{T}_{1}^{T})$.
$\hfill\square$

\vspace{4mm}


Moreover, we introduce two lemmas below to deal with the remainder term $R_{1n}(m)$ in Lemma \ref{lem3.1}.

\begin{lem} \label{lemB1}
Suppose Assumptions \ref{asm2.1}, \ref{asm2.2}(i) and \ref{asm2.3}-\ref{asm2.5} hold. Then, under
$H_{0}$, $n\|R_{1n}(m)\|=o_{p}(1)$, where $R_{1n}(m)$ is defined as in
(\ref{A7}).
\end{lem}

\begin{proof}\,\, As for the proof of Lemma \ref{lem3.2}, we only prove the result for $m=0$. Rewrite
$
R_{1n}(0)=R_{n}^{(1)}+R_{n}^{(2)}+R_{n}^{(3)}+R_{n}^{(4)}
$,
where
$$R_{n}^{(d)}=\frac{1}{n^{2}}\sum_{i,j}R_{ijij}^{(d)}+\frac{1}{n^{4}}\sum_{i,j,q,r}R_{ijqr}^{(d)}
-\frac{2}{n^{3}}\sum_{i,j,q}R_{ijiq}^{(d)}$$
for $d=1,2,3,4$, and $R_{ijqr}^{(d)}$ is defined as in (\ref{A5}).

We first consider $R_{n}^{(1)}$. By (\ref{A3})-(\ref{A4}), we can rewrite $R_{ijqr}^{(1)}$ as
\begin{flalign}
R_{ijqr}^{(1)}&=[\overline{R}_{ijqr}^{(2)}]^{T}(H_{ijqr}^{\dag}-H_{ijqr})
\overline{R}_{ijqr}^{(2)}\nonumber\\
&\quad+[\overline{R}_{ijqr}^{(3)}]^{T}(H_{ijqr}^{\dag}-H_{ijqr})
\overline{R}_{ijqr}^{(3)}\nonumber\\
&\quad+\left[(\widehat{\theta}_{n}-\theta_{0})^{T}\frac{\p G_{ijqr}(\theta_{0})}{\p\theta}\right](H_{ijqr}^{\dag}-H_{ijqr})
\left[\frac{\p G_{ijqr}(\theta_{0})}{\p\theta^{T}}(\widehat{\theta}_{n}-\theta_{0})\right]\nonumber\\
&\quad+2[\overline{R}_{ijqr}^{(2)}]^{T}(H_{ijqr}^{\dag}-H_{ijqr})
\overline{R}_{ijqr}^{(3)}\nonumber\\
&\quad+
2[\overline{R}_{ijqr}^{(2)}]^{T}(H_{ijqr}^{\dag}-H_{ijqr})
\left[\frac{\p G_{ijqr}(\theta_{0})}{\p\theta^{T}}(\widehat{\theta}_{n}-\theta_{0})\right]\nonumber\\
&\quad+2[\overline{R}_{ijqr}^{(3)}]^{T}(H_{ijqr}^{\dag}-H_{ijqr})
\left[\frac{\p G_{ijqr}(\theta_{0})}{\p\theta^{T}}(\widehat{\theta}_{n}-\theta_{0})\right]\nonumber\\
&=: r_{1,ijqr}^{(1)}+r_{2,ijqr}^{(1)}+r_{3,ijqr}^{(1)}+r_{4,ijqr}^{(1)}+r_{5,ijqr}^{(1)}
+r_{6,ijqr}^{(1)}.\label{A10}
\end{flalign}
Then, by (\ref{A10}), we have $R_{n}^{(1)}=\sum_{d=1}^{6}\Delta_{d}^{(1)}$, where
$$\Delta_{d}^{(1)}=\frac{1}{n^{2}}\sum_{i,j}r_{d,ijij}^{(1)}+\frac{1}{n^{4}}\sum_{i,j,q,r}r_{d,ijqr}^{(1)}
-\frac{2}{n^{3}}\sum_{i,j,q}r_{d,ijiq}^{(1)}.$$

For the first entry of
$[H_{ijqr}^{\dag}-H_{ijqr}]$, we have
$\big\|k_{xx}(\widehat{\eta}_{1i}^{\dag},\widehat{\eta}_{1j}^{\dag})l(\widehat{\eta}_{2q}^{\dag},\widehat{\eta}_{2r}^{\dag})-
k_{xx}(\eta_{1i},\eta_{1j})$ $l(\eta_{2q},\eta_{2r})\big\|
\leq C
\big[\|\widehat{\eta}_{1i}^{\dag}-\eta_{1i}\|
+\|\widehat{\eta}_{1j}^{\dag}-\eta_{1j}\|
+\|\widehat{\eta}_{2q}^{\dag}-\eta_{2q}\|
+\|\widehat{\eta}_{2r}^{\dag}-\eta_{2r}\|\big]$
by Triangle's inequality and Assumption \ref{asm2.5}.
Meanwhile, by Taylor's expansion and Assumptions \ref{asm2.2}(i) and \ref{asm2.3}, we can show that
$\|\widehat{\eta}_{st}^{\dag}-\eta_{1i}\|\leq \|\widehat{R}_{st}(\widehat{\theta}_{sn})\|+\|\widehat{\theta}_{sn}-\theta_{s0}\|\sup_{\theta_{s}}
\big\|\frac{\p g_{st}(\theta_{s})}{\p\theta_{s}}\big\|=\|\widehat{R}_{st}(\widehat{\theta}_{sn})\|+o_{p}(1)$, where
$\widehat{R}_{st}(\theta_{s})$ is defined as in Assumption \ref{asm2.4},
$o_{p}(1)$ holds uniformly in $t$ due to the fact that $\sqrt{n}\|\widehat{\theta}_{sn}-\theta_{s0}\|=O_{p}(1)$ and
\begin{flalign}
\frac{1}{\sqrt{n}}\max_{1\leq t\leq n}\sup_{\theta_{s}}\left\|\frac{\p g_{st}(\theta_{s})}{\p\theta_{s}}\right\|=o_{p}(1)\label{A11}
\end{flalign}
by Assumption \ref{asm2.2}(i). Hence, it follows that
\begin{flalign}
&\big\|k_{xx}(\widehat{\eta}_{1i}^{\dag},\widehat{\eta}_{1j}^{\dag})l(\widehat{\eta}_{2q}^{\dag},\widehat{\eta}_{2r}^{\dag})-
k_{xx}(\eta_{1i},\eta_{1j})l(\eta_{2q},\eta_{2r})\big\|\nonumber\\
&\quad\leq C\left[\|\widehat{R}_{1i}(\widehat{\theta}_{1n})\|
+\|\widehat{R}_{1j}(\widehat{\theta}_{1n})\|+\|\widehat{R}_{2q}(\widehat{\theta}_{2n})\|+\|\widehat{R}_{2r}(\widehat{\theta}_{2n})\|\right]\nonumber\\
&\quad\quad+o_{p}(1), \label{A12}
\end{flalign}
where $o_{p}(1)$ holds uniformly in $i,j,q,r$. Similarly, (\ref{A12}) holds for other entries of
$[H_{ijqr}^{\dag}-H_{ijqr}]$. Note that
\begin{flalign}
\|\overline{R}_{ijqr}^{(2)}\|\leq\|\widehat{R}_{1i}(\widehat{\theta}_{1n})\|
+\|\widehat{R}_{1j}(\widehat{\theta}_{1n})\|+\|\widehat{R}_{2q}(\widehat{\theta}_{2n})\|+\|\widehat{R}_{2r}(\widehat{\theta}_{2n})\|.\label{A13}
\end{flalign}
Using the inequality $(\sum_{d=1}^{4}|a_{d}|)^{3}\leq
C\sum_{d=1}^{4}|a_{d}|^{3}$, by Assumption \ref{asm2.4} and (\ref{A12})-(\ref{A13}),  it is not hard to show that
\begin{flalign}
n\|\Delta_{1}^{(1)}\|=O_{p}(1/n).\label{A14}
\end{flalign}
Furthermore, by Taylor's expansion, Assumptions \ref{asm2.2}(i) and \ref{asm2.3}, and a similar argument as for (\ref{A11}), it is straightforward to see that
\begin{flalign*}
\|\overline{R}_{ijqr}^{(3)}\|&\leq\left\|\frac{\p G_{ijqr}(\theta^{\dag})}{\p\theta^{T}}-\frac{\p G_{ijqr}(\theta_{0})}{\p\theta^{T}}\right\|\times\|\widehat{\theta}_{n}-\theta_{0}\|\\
&\leq\left[2\max_{1\leq t\leq n}\sup_{\theta_{1}}\left\|\frac{\p^{2}g_{1t}(\theta_{1})}{\p\theta_{1}^{2}}\right\|
+2\max_{1\leq t\leq n}\sup_{\theta_{2}}\left\|\frac{\p^{2}g_{2t}(\theta_{2})}{\p\theta_{2}^{2}}\right\|
\right]\times\|\widehat{\theta}_{n}-\theta_{0}\|^{2}\\
&=o_{p}(1/\sqrt{n}),
\end{flalign*}
where $o_{p}(1)$ holds uniformly in $i,j,q,r$.
As for (\ref{A14}), it entails that
$
n\|\Delta_{2}^{(1)}\|=o_{p}(1).
$
Similarly, we can show that
$
n\|\Delta_{d}^{(1)}\|=o_{p}(1)
$
for $d=3,4,5,6$. Therefore, it follows that $n\|R_{n}^{(1)}\|=o_{p}(1)$.
By the analogous arguments, we can also show that $n\|R_{n}^{(d)}\|=o_{p}(1)$ for $d=3,4$.

Next, we consider the remaining term $R_{n}^{(2)}$.
Denote $r_{1,ijqr}^{(2)}:=[\overline{R}_{ijqr}^{(2)}]^{T}W_{ijqr}$ and $r_{2,ijqr}^{(2)}:=[\overline{R}_{ijqr}^{(3)}]^{T}W_{ijqr}$.
Then, we can rewrite
$
R_{n}^{(2)}=\Delta_{1}^{(2)}+\Delta_{2}^{(2)}
$,
where
\begin{flalign*}
\Delta_{d}^{(2)}=\frac{1}{n^{2}}\sum_{i,j}r_{d,ijij}^{(2)}+\frac{1}{n^{4}}\sum_{i,j,q,r}r_{d,ijqr}^{(2)}
-\frac{2}{n^{3}}\sum_{i,j,q}r_{d,ijiq}^{(2)}
\end{flalign*}
for $d=1, 2$. By Assumptions \ref{asm2.2}(i) and \ref{asm2.3}-\ref{asm2.5} and (\ref{A13}), we have
$
n\|\Delta_{1}^{(2)}\|=O_{p}(1/n).
$
Rewrite
$
\Delta_{2}^{(2)}=(\widehat{\theta}_{1n}-\theta_{10})^{T}\Delta_{21}^{(2)}+(\widehat{\theta}_{2n}-\theta_{20})^{T}\Delta_{22}^{(2)},
$
where
\begin{flalign*}
\Delta_{2d}^{(2)}=\frac{1}{n^{2}}\sum_{i,j}r_{2d,ijij}^{(2)}+\frac{1}{n^{4}}\sum_{i,j,q,r}r_{2d,ijqr}^{(2)}
-\frac{2}{n^{3}}\sum_{i,j,q}r_{2d,ijiq}^{(2)}
\end{flalign*}
for $d=1,2$, with $r_{21,ijqr}^{(2)}=k_{ij}^{\dag}l_{qr}$ and $r_{22,ijqr}^{(2)}=k_{ij}l_{qr}^{\dag}$. Here,
\begin{flalign*}
k_{ij}^{\dag}&=\left[\frac{\p
g_{1i}(\theta_{1}^{\dag})}{\p\theta_{1}}-\frac{\p
g_{1i}(\theta_{10})}{\p\theta_{1}}\right]k_{x}(\eta_{1i},\eta_{1j})+
\left[\frac{\p
g_{1j}(\theta_{1}^{\dag})}{\p\theta_{1}}-\frac{\p
g_{1j}(\theta_{10})}{\p\theta_{1}}\right]k_{y}(\eta_{1i},\eta_{1j}),\\
l_{qr}^{\dag}&=\left[\frac{\p
g_{2q}(\theta_{2}^{\dag})}{\p\theta_{2}}-\frac{\p
g_{2q}(\theta_{20})}{\p\theta_{2}}\right]l_{x}(\eta_{2q},\eta_{2r})+
\left[\frac{\p
g_{2r}(\theta_{2}^{\dag})}{\p\theta_{2}}-\frac{\p
g_{2r}(\theta_{20})}{\p\theta_{2}}\right]l_{y}(\eta_{2q},\eta_{2r}).
\end{flalign*}

By the mean value theorem, $k_{ij}^{\dag}=(\theta^{\dag}_{1}-\theta_{10})^{T}k_{ij}^{\S}$, where
$k_{ij}^{\S}$ is defined explicitly, and
it satisfies that
\begin{flalign}
\Delta_{21}^{(2)\S}&:=\frac{1}{n^{2}}\sum_{i,j}k_{ij}^{\S}l_{ij}+\frac{1}{n^{4}}\sum_{i,j,q,r}k_{ij}^{\S}l_{qr}
-\frac{2}{n^{3}}\sum_{i,j,q}k_{ij}^{\S}l_{iq}=O_{p}(1/n).\label{A15}
\end{flalign}
Here, (\ref{A15}) holds, since $\Delta_{21}^{(2)\S}$ under $H_{0}$ is a
 degenerate $V$-statistic by Assumptions \ref{asm2.1} and \ref{asm2.5} and a similar argument as for Lemma \ref{lem3.2}(ii). Note that $\Delta_{21}^{(2)}=(\theta^{\dag}_{1}-\theta_{10})^{T}\Delta_{21}^{(2)\S}$ and
 $\|\theta^{\dag}_{1}-\theta_{10}\|\leq \|\widehat{\theta}_{1n}-\theta_{10}\|=o_{p}(1)$.
 Therefore, it follows that $\sqrt{n}\|\Delta_{21}^{(2)}\|=o_{p}(1)$.
 Similarly, we can show that
 $\sqrt{n}\|\Delta_{22}^{(2)}\|=o_{p}(1)$, and it follows that $n\|R_{n}^{(2)}\|=o_{p}(1)$.
 This completes the proof.
\end{proof}

\begin{lem}\label{lemB2}
Suppose  Assumptions \ref{asm2.1}-\ref{asm2.5} hold. Then, $\sqrt{n}\|R_{n}(m)\|=o_{p}(1)$, where $R_{n}(m)$ is defined as in (\ref{A7}).
\end{lem}

\begin{proof}\,\,The proof is the same as the one for Lemma \ref{lemB1}, except that  when $H_{0}$ does not hold,
we can only have $\Delta_{21}^{(2)\S}=O_{p}(1)$ in (\ref{A15}) by Assumption \ref{asm2.2}(ii) and part (c) of Theorem 1 in Denker and Keller (1983).
\end{proof}

\textsc{Proof of Theorem \ref{thm3.1}.}\,\,(i) By Lemmas \ref{lem3.1} and \ref{lemB1},
$$N[S_{1n}(m)]=Z_{1n}(m)+o_{p}(1),$$
 where
\begin{flalign*}
Z_{1n}(m)&:=N[S_{1n}^{(0)}(m)]+\zeta_{1n}^{T}[NS_{1n}^{(11)}(m)]+\zeta_{2n}^{T}[NS_{1n}^{(12)}(m)]\\
&\quad+
\frac{1}{2}\zeta_{1n}^{T}[NS_{1n}^{(21)}(m)]\zeta_{1n}+\frac{1}{2}\zeta_{2n}^{T}[NS_{1n}^{(22)}(m)]\zeta_{2n}\\
&\quad+
[\sqrt{N}\zeta_{1n}]^{T}S_{1n}^{(23)}(m)[\sqrt{N}\zeta_{2n}].
\end{flalign*}
For $a, b=1, 2$, $S_{1n}^{(ab)}(m)$ is a degenerate V-statistic of order 1 by Lemma \ref{lem3.2}(ii),  and hence
$NS_{1n}^{(ab)}(m)=O_{p}(1)$. By Assumption \ref{asm2.3}, it follows that
\begin{flalign*}
Z_{1n}(m)&=N[S_{1n}^{(0)}(m)]+
[\sqrt{N}\zeta_{1n}]^{T}S_{1n}^{(23)}(m)[\sqrt{N}\zeta_{2n}]+o_{p}(1)\\
&=N[S_{1n}^{(0)}(m)]+
[\sqrt{N}\zeta_{1n}]^{T}\Lambda^{(23)}_{m}[\sqrt{N}\zeta_{2n}]+o_{p}(1),
\end{flalign*}
where the last equality holds by the law of large numbers for V-statistics.
Hence, $Z_{1n}(m)\to_{d}\chi_{m}$ as $n\to\infty$ by (\ref{3.9}), Lemma \ref{lem3.3}, and the continuous mapping theorem.
This completes the proof of (i).

(ii) It follows by a similar argument as for (i).
$\hfill\square$

\vspace{4mm}


\textsc{Proof of Theorem \ref{thm3.2}.}\,\,(i) By Lemmas \ref{lem3.1} and \ref{lemB2}, we have
\begin{flalign}\label{A17}
\sqrt{N}\left[S_{1n}(m)-\Lambda^{(0)}_{m}\right]=
\overline{Z}_{1n}(m)+o_{p}(1),
\end{flalign}
where $\Lambda^{(0)}_{m}=E[h^{(0)}_{m}(\eta_{1}^{(m)},\eta_{2}^{(m)},\eta_{3}^{(m)},\eta_{4}^{(m)})]>0$ and
\begin{flalign*}
\overline{Z}_{1n}(m)&:=\sqrt{N}\left[S_{1n}^{(0)}(m)-\Lambda^{(0)}_{m}\right]+[\sqrt{N}\zeta_{1n}]^{T}S_{1n}^{(11)}(m)
+[\sqrt{N}\zeta_{2n}]^{T}S_{1n}^{(12)}(m)\\
&\quad+
\frac{1}{2\sqrt{N}}\left\{[\sqrt{N}\zeta_{1n}]^{T}S_{1n}^{(21)}(m)[\sqrt{N}\zeta_{1n}]+[\sqrt{N}\zeta_{2n}]^{T}S_{1n}^{(22)}(m)
[\sqrt{N}\zeta_{2n}]\right.\\
&\quad\left.+
2[\sqrt{N}\zeta_{1n}]^{T}S_{1n}^{(23)}(m)[\sqrt{N}\zeta_{2n}]\right\}.
\end{flalign*}

First, since $S_{1n}^{(0)}(m)$ is a non-degenerate $V$-statistic under $H_{1}^{(m)}$,  part (c) of Theorem 1 in Denker and Keller (1983) implies that
\begin{flalign}
\sqrt{N}\left[S_{1n}^{(0)}(m)-\Lambda^{(0)}_{m}\right]=O_{p}(1).\label{A18}
\end{flalign}
Second, by the law of large numbers for V-statistics and Assumption \ref{asm2.3}, it follows that
\begin{flalign}
&[\sqrt{N}\zeta_{1n}]^{T}S_{1n}^{(11)}(m)=
\left[\frac{1}{\sqrt{n}}\sum_{i=1}^{n}\pi_{1i}\right]^{T}\Lambda^{(11)}_{m}+o_{p}(1)=O_{p}(1),\label{A19}\\
&[\sqrt{N}\zeta_{2n}]^{T}S_{1n}^{(12)}(m)=
\left[\frac{1}{\sqrt{n}}\sum_{i=1}^{n}\pi_{2i}\right]^{T}\Lambda^{(12)}_{m}+o_{p}(1)=O_{p}(1),\label{A20}\\
&\frac{1}{2\sqrt{N}}\left\{[\sqrt{N}\zeta_{1n}]^{T}S_{1n}^{(21)}(m)[\sqrt{N}\zeta_{1n}]+[\sqrt{N}\zeta_{2n}]^{T}S_{1n}^{(22)}(m)
[\sqrt{N}\zeta_{2n}]\right.\nonumber\\
&\quad\left.+
2[\sqrt{N}\zeta_{1n}]^{T}S_{1n}^{(23)}(m)[\sqrt{N}\zeta_{2n}]\right\}=o_{p}(1),\label{A21}
\end{flalign}
where $\Lambda^{(1s)}_{m}=E[h^{(1s)}_{m}(\eta_{1}^{(m)},\eta_{2}^{(m)},\eta_{3}^{(m)},\eta_{4}^{(m)})]$ for $s=1, 2$.
By (\ref{A18})-(\ref{A21}), $\overline{Z}_{1n}(m)=O_{p}(1)$, which together with (\ref{A17}) implies that $n[S_{1n}(m)]\to\infty$ in probability as $n\to\infty$.
This completes the proof of (i).

(ii) It follows by a similar argument as for (i). $\hfill\square$

\vspace{4mm}

Let $\varsigma_{st}=\big(\eta_{st},\frac{\p g_{st}(\theta_{s0})}{\p\theta_{s}}\big)$ for $s=1, 2$.
To prove Theorem \ref{thm4.1}, we need the following two lemmas, where the first lemma provides some useful results to prove the second one.


\begin{lem}\label{lemB3}
Suppose Assumptions \ref{asm2.1}, \ref{asm2.2}(i) and \ref{asm2.3}-\ref{asm2.5} hold. Then, under $H_{0}$, for $\forall K_{0}>0$,
\begin{flalign*}
(i)\,\,   \sup_{\Omega_{1}}&\left|\frac{1}{N^4}\sum_{q,q',r,r'}h^{(0)}_{m}\left(
x_{1}, x_{2},\left(\eta_{1q}, \eta_{2q'+m}\right),\left(
\eta_{1r}, \eta_{2r'+m}\right)\right)\right.\\
&\,\,\,\,\left.
-h_{2m}^{(0)}(x_1,x_2)\right|=o_p(1),
\end{flalign*}
where $\Omega_{1}=\{(x_{1}, x_{2}): \|x_{s}\|\leq K_{0}\mbox{ for }s=1, 2\}$, and $h_{2m}^{(0)}(x_1,x_2)$ is defined as in (\ref{h2m0});
\begin{flalign*}
(ii)\,\,   \sup_{\Omega_{2}}&\left|\frac{1}{N^4}\sum_{i',j',q',r'}h^{(23)}_{m}\left(\left(
z_{11},
\varsigma_{2i'+m}\right),\left(
z_{12},
\varsigma_{2j'+m}
\right),\left(
z_{13},
\varsigma_{2q'+m}
\right),\left(
z_{14},
\varsigma_{2r'+m}
\right)\right)\right.\\
&\,\,\, \,\left.-E\left[h^{(23)}_{m}\left(\left(
z_{11},
\varsigma_{21}
\right),\left(
z_{12},
\varsigma_{22}
\right),\left(
z_{13},
\varsigma_{23}
\right),\left(
z_{14},
\varsigma_{24}
\right)\right)\right]\right|=o_p(1),
\end{flalign*}
where $\Omega_{2}=\{(z_{11},z_{12},z_{13},z_{14}):\|z_{1s}\|\leq K_{0} \mbox{ for } s=1,2,3,4\}$;
\begin{flalign*}
(iii)\,\,   \sup_{\Omega_{3}}&\left|\frac{1}{N^4}\sum_{i,j,q,r}h^{(23)}_{m}\left(\left(
\varsigma_{1i},
z_{21}
\right),\left(
\varsigma_{1j},
z_{22}
\right),\left(
\varsigma_{1q},
z_{23}
\right),\left(
\varsigma_{1r},
z_{24}
\right)\right)\right.\\
&\,\,\, \,\left.-E\left[h^{(23)}_{m}\left(\left(
\varsigma_{11},
z_{21}
\right),\left(
\varsigma_{12},
z_{22}
\right),\left(
\varsigma_{13},
z_{23}
\right),\left(
\varsigma_{14},
z_{24}
\right)\right)\right]\right|=o_p(1),
\end{flalign*}
where  $\Omega_{3}=\{(z_{21},z_{22},z_{23},z_{24}):\|z_{2s}\|\leq K_{0} \mbox{ for } s=1,2,3,4\}$.
\end{lem}

\begin{proof} $(i)$ Denote $x_{1}=(x_{11},x_{21})$ and $x_{2}=(x_{12},x_{22})$. Without loss of generality, we assume that $m=0$. By the definition of $h^{(00)}_{0}$, it has 24 different terms, and we
only give the proof for its first term. That is, we are going to show that
\begin{flalign}
&\frac{1}{N^4}\sum_{q,r,q',r'}\tilde{k}_{12}^{(0)}[\tilde{l}_{12}^{(0)}+l_{q'r'}^{(0)}-2\tilde{l}_{1 q'}^{(0)}-E(\tilde{l}_{12}^{(0)})-E(l_{34}^{(0)})+2E(\tilde{l}_{13}^{(0)})]=o_{p}(1),\label{B15}
\end{flalign}
where $o_{p}(1)$ holds uniformly in $\Omega_{1}$, $\tilde{k}_{12}^{(0)}=k(x_{11},x_{12})$, $\tilde{l}_{12}^{(0)}=k(x_{21},x_{22})$, $l_{q'r'}^{(0)}=k(\eta_{2q'},\eta_{2r'})$,
$\tilde{l}_{1 q'}^{(0)}=k(x_{21},\eta_{2q'})$, $l_{34}^{(0)}=k(\eta_{23},\eta_{24})$, and $\tilde{l}_{1 3}^{(0)}=k(x_{21},\eta_{23})$.

By the triangle's inequality, we have
\begin{flalign*}
&\left|\frac{1}{N^4}\sum_{q,r,q',r'}\tilde{k}_{12}^{(0)}[\tilde{l}_{12}^{(0)}+l_{q'r'}^{(0)}-2\tilde{l}_{1 q'}^{(0)}-E(\tilde{l}_{12}^{(0)})-E(l_{34}^{(0)})+2E(\tilde{l}_{13}^{(0)})]\right|\\
& = \left|\frac{\tilde{k}_{12}^{(0)}}{N^4}\sum_{q,r,q',r'}[l_{q'r'}^{(0)}-2\tilde{l}_{1 q'}^{(0)}-E(l_{34}^{(0)})+2E(\tilde{l}_{13}^{(0)})]\right|\\
&\leq \left|\f{C}{N^2}\sum_{q',r'=1}^n[l_{q'r'}^{(0)}-E(l_{34}^{(0)})]\right|
+\left|\f{C}{N}\sum_{q'=1}^n[\tilde{l}_{1q'}^{(0)}-E(\tilde{l}_{13}^{(0)})]\right|.
\end{flalign*}
Hence, it follows that (\ref{B15}) holds by noting the fact that
\begin{flalign}
&\f{1}{N^2}\sum_{q',r'=1}^n[l_{q'r'}^{(0)}-E(l_{34}^{(0)})]=o_{p}(1),\label{B16}\\
&\sup_{\Omega_{1}}\f{1}{N}\sum_{q'=1}^n[\tilde{l}_{1q'}^{(0)}-E(\tilde{l}_{13}^{(0)})]=o_{p}(1),\label{B17}
\end{flalign}
where (\ref{B16}) holds by the law of large numbers for V-statistics,
and (\ref{B17}) holds by Assumption \ref{asm2.5} and standard arguments for uniform convergence.

$(ii)$ \& $(iii)$ The conclusions hold by similar arguments as for $(i)$.
\end{proof}

\begin{lem}\label{lemB4}
Suppose Assumptions \ref{asm2.1}, \ref{asm2.2}(i) and \ref{asm2.3}-\ref{asm2.5} hold. Then, under $H_{0}$,
\begin{flalign*}
(i)\,\, &  \sup_{\Omega_{1}}\left|h_{2m}^{(0*)}(x_1,x_2)
-h_{2m}^{(0)}(x_1,x_2)\right|=o_{p}(1),
\end{flalign*}
where $\Omega_{1}$, $h_{2m}^{(0)}(x_1,x_2)$ and $h_{2m}^{(0*)}(x_1,x_2)$ are defined as in Lemma \ref{lemB3}(i), (\ref{h2m0}) and (\ref{h2m0_star}), respectively;
\begin{flalign*}
(ii)\,\, &\left|\Lambda_{m}^{(23*)}-\Lambda_{m}^{(23)}\right|=o_{p}(1),
\end{flalign*}
where $\Lambda_{m}^{(23)}$ and $\Lambda_{m}^{(23*)}$ are defined as in (\ref{lambda_m23}) and (\ref{lambda_m23_star}), respectively.
\end{lem}


\begin{proof}
(i) First, it is straightforward to see that
\begin{flalign}
h_{2m}^{(0*)}(x_1,x_2)&=
\frac{1}{N^4}\sum_{q,q',r,r'}h^{(0)}_{m}\left(x_{1}, x_{2},\left(
\widehat{\eta}_{1q},
\widehat{\eta}_{2q'+m}
\right),\left(
\widehat{\eta}_{1r},
\widehat{\eta}_{2r'+m}
\right)\right)\nonumber\\
&=
\frac{1}{N^4}\sum_{q,q',r,r'}h^{(0)}_{m}\left(x_{1}, x_{2},\left(
\eta_{1q},
\eta_{2q'+m}
\right),\left(
\eta_{1r},
\eta_{2r'+m}
\right)\right)+o_{p}(1),\label{B18}
\end{flalign}
where $o_{p}(1)$ holds uniformly in $\Omega_{1}$ by Taylor's expansion and Assumptions \ref{asm2.3} and \ref{asm2.5}.
Then, the conclusion holds by (\ref{B18}) and Lemma \ref{lemB3}(i).

(ii) Define
$$\mathcal{H}(i,i',j,j',q,q',r,r')=h^{(23)}_{m}\left(\left(
\varsigma_{1i},
\varsigma_{2i'}
\right),\left(
\varsigma_{1j},
\varsigma_{2j'}
\right),\left(
\varsigma_{1q},
\varsigma_{2q'}
\right),\left(
\varsigma_{1r},
\varsigma_{2r'}
\right)\right).$$
By a similar argument as for (\ref{B18}), we have
\begin{flalign*}
\Lambda_{m}^{(23*)}-\Lambda_{m}^{(23)}=\Xi_{0}+o_{p}(1),
\end{flalign*}
where
$$\Xi_{0}=\f{1}{N^8}\sum_{i,j,q,r,i',j',q',r'}\mathcal{H}(i,i'+m,j,j'+m,q,q'+m,r,r'+m)-\Lambda_{m}^{(23)}.$$
Rewrite
\begin{flalign}
\Xi_{0}:=\f{1}{N^4}\sum_{i,j,q,r}\Xi_{1,ijqr}+\f{1}{N^4}\sum_{i,j,q,r}\Xi_{2,ijqr},\label{B19}
\end{flalign}
where
\begin{flalign*}
&\Xi_{1,ijqr}=\f{1}{N^4}\sum_{i',j',q',r'}\mathcal{H}(i,i'+m,j,j'+m,q,q'+m,r,r'+m)\\
&\, \, \,\,\,\,\, \, \,\,\,\,\, \, \,\,\,\,\, \, \,\,\,\,\, \, \,\,\,\, -E_{\varsigma_{21},\varsigma_{22},\varsigma_{23},\varsigma_{24}}\left[\mathcal{H}(i,1,j,2,q,3,r,4)\right],\\
&\Xi_{2,ijqr}=E_{\varsigma_{21},\varsigma_{22},\varsigma_{23},\varsigma_{24}}\left[\mathcal{H}(i,1,j,2,q,3,r,4)\right]-\Lambda_{m}^{(23)}.
\end{flalign*}
By Lemma \ref{lemB3}(ii), $\Xi_{1,ijqr}=o_{p}(1)$ uniformly in $i,j,q,r$, and hence
\begin{flalign}
\f{1}{N^4}\sum_{i,j,q,r}\Xi_{1,ijqr}=o_{p}(1).\label{B20}
\end{flalign}
Moreover,
we can rewrite
\begin{flalign}\label{B21}
\f{1}{N^4}\sum_{i,j,q,r}\Xi_{2,ijqr}&=E_{\varsigma_{21},\varsigma_{22},\varsigma_{23},\varsigma_{24}}
\big\{\mathcal{H}(i,1,j,2,q,3,r,4) \nonumber\\
&\quad-E_{\varsigma_{11},\varsigma_{12},\varsigma_{13},\varsigma_{14}}
\left[\mathcal{H}(1,1,2,2,3,3,4,4)\right]\big\},
\end{flalign}
where we have used the fact that under $H_{0}$,
$$\Lambda_{m}^{(23)}=
E_{\varsigma_{21},\varsigma_{22},\varsigma_{23},\varsigma_{24}}
E_{\varsigma_{11},\varsigma_{12},\varsigma_{13},\varsigma_{14}}\left[\mathcal{H}(1,1,2,2,3,3,4,4)\right].$$
By (\ref{B21}), Lemma \ref{lemB3}(iii), Assumptions \ref{asm2.2}(i) and \ref{asm2.5}, and the dominated convergence theorem, we can show that
\begin{flalign}
\f{1}{N^4}\sum_{i,j,q,r}\Xi_{2,ijqr}=o_{p}(1).\label{B22}
\end{flalign}
Hence, the conclusion holds by (\ref{B19})-(\ref{B20}) and (\ref{B22}).
\end{proof}

\textsc{Proof of Theorem \ref{thm4.1}.}\,\, (i)
By Assumptions \ref{asm4.1} and \ref{asm4.2}(i), $\sqrt{N}\zeta_{sn}^{*}=O_{p}^{*}(1)$.
Then, by (\ref{4.1})-(\ref{4.2}), Assumption \ref{asm4.2}, and a similar argument as for Lemmas \ref{lem3.2}(ii)-(iii) and \ref{lemB1},
we can show that
\begin{flalign}
S_{1n}^{**}(m)&=
\sum_{j=1}^{\infty}\lambda_{jm}^*\left[\frac{1}{\sqrt{N}}\sum_{i=1}^N\Phi_{jm}^*(\widehat{\eta}_i^{(m*)})\right]\nonumber\\
&\quad
+[\sqrt{N}\zeta_{1n}^{*T}]\Lambda^{(23*)}_{m}[\sqrt{N}\zeta_{2n}^*]+o_{p}^{*}(1)=O_{p}^{*}(1).\label{AA8}
\end{flalign}
This completes the proof of (i).

(ii) It follows by a similar argument as for (i).

(iii) Let $\mathcal{T}_{1i}^{*}=\left(\big(\Phi_{jm}^{*}(\widehat{\eta}_{i}^{(m*)})\big)_{j\geq1, 0\leq m\leq M}\right)^{T}$,  $\mathcal{T}_{2i}^{*}=\big((\pi_{si}^{*T})_{1\leq s\leq 2}\big)^{T}$, and
$$\mathcal{T}_{n}^*=\big(\frac{1}{\sqrt{N}}\sum_{i=1}^{N}\mathcal{T}_{1i}^{*T}, \frac{1}{\sqrt{n}}\sum_{i=1}^{n}\mathcal{T}_{2i}^{*T}\big)^{T},$$
where $\pi_{si}^*$ is defined as in Assumption \ref{asm4.1}.
Also, let $\mathcal{T}_{i}^{*}=(\mathcal{T}_{1i}^{*T},\mathcal{T}_{2i}^{*T})^{T}$. As for Lemma \ref{lem3.3}, it is not hard to see that
conditional on $\varpi_{n}$,
\begin{equation}\label{A27}
\mathcal{T}_n^*\rightarrow_d \mathcal{T}^*
\end{equation}
in probability as $n\to\infty$, where  $\mathcal{T}^*$ is a multivariate normal distribution with covariance matrix $\overline{\mathcal{T}}^{*}$, and
$\overline{\mathcal{T}}^{*}=\lim_{n\to\infty}E^{*}(\mathcal{T}_{1}^{*}\mathcal{T}_{1}^{*T})=
E(\mathcal{T}_1\mathcal{T}_1^T)=\overline{\mathcal{T}}$ in probability
by Assumption \ref{asm4.2}.

Next, by Lemma \ref{lemB4}(i) and  Corollary XI.9.4(a) in Dunford and Schwartz (1963, p.1090), we can get
\begin{equation}\label{AA10}
|\lambda_{jm}^*-\lambda_{jm}|=o(1).
\end{equation}
Hence, the conclusion holds by (\ref{AA8})-(\ref{AA10}), Lemma \ref{lemB4}(ii), and the continuous mapping theorem.
This completes the proof of (iii).

(iv) It follows by a similar argument as for (iii).  $\hfill\square$

%




\end{document}